\documentclass{llncs}

\usepackage{times}
\usepackage[linesnumbered,lined,vlined]{algorithm2e}
\usepackage{graphicx}
\usepackage{gastex}
\usepackage{wrapfig}
\usepackage{paralist}
\usepackage{multirow}
\usepackage{paralist}
\usepackage[T1]{fontenc}
\usepackage{mathtools}
\usepackage{xcolor}
\usepackage{caption}
\usepackage{subcaption}
\usepackage{wrapfig}
\usepackage{floatrow}
\usepackage{multirow}
\usepackage{url}

\newfloatcommand{capbtabbox}{table}[][\FBwidth]

\usepackage{blindtext}
\usepackage{listings}
\lstdefinelanguage{program}{%
  keywords={%
    atomic,%
    assume,assert,call,return,new,%
    false,true,duplicate,restart,lock,unlock,%
    locate,insert,delete,contains,removeRight,rotateRightLeft,
    rcu_read_lock,rcu_read_unlock,synchronize_rcu%
  },  %
  morecomment=[l]{//},
  morecomment=[s]{/*}{*/},
  morecomment=[n]{(**}{**)},
  mathescape=true,
  escapeinside=`',
}
\usepackage{amssymb}
\usepackage{amsmath}
\usepackage{url}
\usepackage{xcolor}

\newcommand{\ignore}[1]{}

\spnewtheorem*{notation}{Notation}{\itshape}{\rmfamily}

\newcommand{\SCfunc}{f}
\newcommand{\Inv}{\mathit{Inv}}

\SetKwFunction{CG}{ConvexGeneralizer}
\SetKwFunction{FG}{FarkasGeneralizer}

\newcommand{\true}{\textit{true}}
\newcommand{\false}{\textit{false}}

\newcommand{\Bad}{\mathit{Bad}}

\newcommand{\prop}[2]{(#1,#2)}
\newcommand{\kmodels}{\models^k}

\newcommand{\Spacer}{\textsc{Spacer}\xspace}

\newcommand{\papercomment}[1]{}

\newcommand{\Vars}{\mathcal{V}}
\newcommand{\var}{v}

\newcommand{\BMC}{\mathit{BMC}}

\newcommand{\Cond}{C}
\newcommand{\TS}{T}
\newcommand{\States}{S}
\newcommand{\state}{s}
\newcommand{\Trans}{R}
\newcommand{\Tr}{\Trans}
\newcommand{\preS}{\textit{pre}}
\newcommand{\postS}{\textit{post}}
\newcommand{\compop}[1]{{#1^{\SCfunc}}}
\newcommand{\compnof}[1]{{#1^{\|k}}}
\newcommand{\TScomp}{\compop{\TS}}
\newcommand{\TScompnof}{\compnof{\TS}}
\newcommand{\Fcomp}{\compnof{\Terminal}}
\newcommand{\Scomp}{\compnof{\States}}
\newcommand{\Fcompnof}{\compnof{\Terminal}}
\newcommand{\Scompnof}{\compnof{\States}}
\newcommand{\Tcompnof}{\compnof{\Trans}}
\newcommand{\TScompprime}{{\TS^{\SCfunc'}}}
\newcommand{\TScompdoubleprime}{{\TS^{\SCfunc''}}}
\newcommand{\scomp}{{\state^{\|}}}
\newcommand{\Tcomp}{\compop{\Trans}}
\newcommand{\Varscomp}{\compnof{\Vars}}
\newcommand{\picomp}{{\pi^{\|}}}
\newcommand{\Terminal}{F}
\newcommand{\Preds}{\mathcal{P}}
\newcommand{\pred}{p}
\newcommand{\A}[1]{A_{#1}}

\newcommand{\AStates}{\hat{\States}} %
\newcommand{\ATr}{\hat{\Trans}} %
\newcommand{\AF}{\hat{\Terminal}} %
\newcommand{\astate}{\hat{s}}

\newcommand{\cex}{\mathit{cex}}

\newcommand{\lang}{\mathcal{L}}

\newcommand{\powerset}[1]{\mathbb{P}(#1)}

\renewcommand{\implies}{\Rightarrow}

\newcommand{\langof}[1]{\lang_{#1}}
\newcommand{\Unreach}{\textit{Unreach}}
\newcommand{\toBool}[1]{#1(\mathcal{B})}
\newcommand{\toPred}[1]{#1(\Preds)}

\newcommand{\Pdsc}{\textsc{Pdsc}\xspace}
\newcommand{\SeaHorn}{\textsc{SeaHorn}\xspace}

\newcommand{\Syn}{\textsc{Synonym}\xspace}

\usepackage{hyperref}
\usepackage[capitalise]{cleveref}

\usepackage{listings}
\lstdefinestyle{CStyle}{
    language=C,
    showtabs=false,
    tabsize=2,
    basicstyle=\footnotesize
}
\usepackage{verbatim}

\newif\iflong
\newif\ifextended

\longfalse
\extendedtrue

\ifextended
\else
\Crefname{algorithm}{Alg.}{Alg.}
\Crefname{lemma}{Lem.}{Lem.}
\Crefname{figure}{Fig.}{Fig.}
\Crefname{definition}{Def.}{Def.}
\Crefname{section}{Sec.}{Sec.}
\Crefname{appendix}{App.}{App.}
\fi

\author{Ron Shemer\inst{1} \and Arie Gurfinkel\inst{2}\and Sharon Shoham\inst{1} \and Yakir Vizel\inst{2}}

\title{Property Directed Self Composition}

\institute{Tel Aviv University \and University of Waterloo \and The Technion}

\begin{document}

\pagestyle{plain}
\maketitle

\begin{abstract}
We address the problem of verifying \emph{$k$-safety properties}: properties that refer to $k$ interacting executions of a program.
A prominent way to verify $k$-safety properties is by \emph{self composition}. In this approach, the problem of checking $k$-safety over the original program is reduced to checking an ``ordinary'' safety property over a program that executes $k$ copies of the original program in some order. The way in which the copies are composed determines how complicated it is to verify the composed program.
We view this composition as provided by a \emph{semantic self composition function} that maps each state of the composed program to the copies that make a move.
Since the ``quality'' of a self composition function is measured by the ability to verify the safety of the composed program,
we formulate the problem of inferring a self composition function together with the inductive invariant needed to verify safety of the composed program, where both are restricted to a given language. %
We develop a \emph{property-directed} inference algorithm that, given a set of predicates, infers composition-invariant pairs expressed by Boolean combinations of the given predicates, or determines that no such pair exists.
We implemented our algorithm and demonstrate that it is able to find self compositions that are beyond reach of existing tools.
\end{abstract}

\section{Introduction}\label{sec:intro}

Many relational properties, such as noninterference~\cite{DBLP:journals/cacm/DenningD77},
determinism~\cite{DBLP:conf/sec/KarimpourIN15}, service level agreements~\cite{DBLP:journals/jcs/ClarksonS10},
and more, can be reduced to the problem of $k$-safety. Namely, reasoning
about $k$ different traces of a program simultaneously. A common approach to
verifying $k$-safety properties is by means of \emph{self composition}, where
the program is composed with $k$ copies of itself~\cite{DBLP:conf/csfw/BartheDR04,DBLP:conf/sas/TerauchiA05}.
A state of the composed program consists of the states of each copy,
and a trace naturally corresponds to $k$ traces of the original program. Therefore, $k$-safety
properties of the original program become ordinary safety properties of the
composition, hence reducing $k$-safety verification to ordinary safety.
This enables reasoning about $k$-safety properties using any of the
existing techniques for safety verification such as Hoare logic~\cite{DBLP:journals/cacm/Hoare69} or
model checking~\cite{DBLP:reference/mc/2018}.

While self composition is sound and complete for $k$-safety, its applicability
is questionable for two main reasons:
\begin{inparaenum}[(i)] \item considering several copies of the program
greatly increases the state space; and
\item the way in which the different copies are composed when reducing the problem
to safety verification affects the complexity of the resulting
self composed program, and as such affects the complexity of verifying it.
\end{inparaenum}
Improving the applicability of self composition has been the topic of many works~\cite{DBLP:conf/fm/BartheCK11,DBLP:conf/pldi/SousaD16,DBLP:conf/esop/EilersMH18,DBLP:conf/sat/GuptaSMB18,DBLP:conf/cav/PickFG18,DBLP:conf/cav/YangVSGM18}.
However, most efforts are focused on compositions that are pre-defined, or only depend on
syntactic similarities.

In this paper, we take a different approach; we build upon the observation that by
choosing the ``right'' composition, the verification can
be greatly simplified by leveraging ``simple'' correlations between the executions.
To that end, we propose an algorithm, called \Pdsc, for inferring a \emph{property directed} self composition.
Our approach uses a \emph{dynamic} composition, where the
composition of the different copies can change during verification, directed at
simplifying the verification of the composed program.

Compositions considered in previous work differ in the order in
which the copies of the program  execute: either synchronously, asynchronously, or in some mix of the
two~\cite{DBLP:conf/fm/ZaksP08,DBLP:conf/lfcs/BartheCK13,DBLP:conf/esop/EilersMH18}.
To allow %
general compositions,
we define a \emph{composition function} that maps every state of the composed program
to the set of copies that are scheduled in the next step.
This determines the order of execution for the different copies, and thus induces the self composed
program.
Unlike most previous works where the composition is pre-defined based on syntactic
rules only, our composition %
is \emph{semantic} as it is defined over
the state of the composed program.

To capture the %
difficulty of verifying the composed program,
we consider verification by means of inferring an inductive invariant,
parameterized by a language for expressing the inductive invariant.
Intuitively, the more expressive the language needs to be,
the more difficult the verification task is.
We then define the problem of inferring a composition function \emph{together} with an inductive invariant for verifying
the safety of the composed program, where both are restricted to a given language.
Note that for a fixed language $\lang$, an inductive invariant may exist for
some composition function but not for another
\ifextended
\footnote{See \Cref{sec:insufficieny} for an example that requires a non-linear inductive invariant with a composition that is based on the control structure but has a linear invariant with another.}
\fi
.
Thus, the restriction to $\lang$
defines a target for the inference algorithm, which is now
directed at finding a composition that admits an inductive invariant in $\lang$.

\begin{example}
To demonstrate our approach, consider the program in \Cref{fig:example}. The program inserts a new
value into an array. We assume that the array $A$ and its length $len$ are ``low''-security variables,
while the inserted value $h$ is ``high''-security. The first loop finds the location in which
$h$ will be inserted. Note that the number of iterations depends on the value of $h$.
Due to that, the second loop executes to ensure that
the output $i$ (which corresponds to the number of iterations) does not leak sensitive
data. As an example, we emphasize that without the second loop, $i$ could leak the
location of $h$ in $A$.
To express the property that $i$ does not leak sensitive data, we use the 2-safety property that
in any two executions, if the inputs $A$ and $len$ are the same, so is the output $i$.

To verify the 2-safety property, consider two copies of the program.
Let the language $\lang$ for verifying the self composition be defined by the predicates
depicted in \Cref{fig:example}. The most natural self composition to consider is a lock-step
composition, where the copies execute synchronously. However, for such a composition
the composed program may reach a state where, for example, $i_1 = i_2 + 1$.
This occurs when the first copy exists the first loop, while the second copy
is still executing it. Since the language cannot express this correlation between
the two copies, no inductive invariant suffices to verify that $i_1 = i_2$ when
the program terminates.

In contrast, when verifying the 2-safety property,
\Pdsc directs its search towards a composition function
for which an inductive invariant in $\lang$ does exist.
As such, it infers the composition function depicted in \Cref{fig:example},
as well as an inductive invariant in $\lang$. The invariant  for this composition implies that $i_1 = i_2$ at every state.

\end{example}

As demonstrated by the example, \Pdsc focuses on logical languages based on
predicate abstraction~\cite{DBLP:conf/cav/GrafS97},
where inductive invariants can be inferred by model checking.
In order to infer a composition function that admits an inductive invariant in $\lang$,
\Pdsc starts from a default composition function, and modifies its definition based
on the reasoning performed by the model checker during verification. As the
composition function is part of the verified model (recall that it is defined over the program state),
different compositions are part
of the state space explored by the model checker. As a result, a key ingredient of
\Pdsc is identifying ``bad'' compositions that prevent it from finding an
inductive invariant in $\lang$. It is important to note that a naive algorithm
that tries all possible composition functions has a time complexity $O(2^{2^{|\Preds|}})$,
where $\Preds$ is the set of predicates considered.
However, integrating the search for a composition function into the model
checking algorithm allows us to reduce the time complexity of the algorithm to
$2^{O(|\Preds|)}$, where we show that the problem is in fact PSPACE-hard.

We implemented \Pdsc using \SeaHorn~\cite{DBLP:conf/cav/GurfinkelKKN15},
Z3~\cite{DBLP:conf/tacas/MouraB08} and \Spacer~\cite{DBLP:conf/cav/KomuravelliGC14} and evaluated it on examples
that demonstrate the need for nontrivial semantic compositions.
Our results clearly show that \Pdsc can solve complex examples by inferring
the required composition, while other tools cannot verify these examples.
We emphasize that for these particular examples, lock-step composition
is not sufficient.
We also evaluated \Pdsc on the
examples from~\cite{DBLP:conf/pldi/SousaD16,DBLP:conf/cav/PickFG18} that are proven with the trivial
lock-step composition.
On these examples, \Pdsc is comparable to state of the art tools.
\begin{figure}[t]
  \centering
\begin{tabular}{ll}
\hspace{-1cm}
\begin{minipage}{0.48\textwidth}
  \lstset{basicstyle=\ttfamily\scriptsize}
  \lstset{language=C,morekeywords={assert,predicates,nd,composition}}
\begin{lstlisting}[mathescape=true]
int arrayInsert(int[] A, int len, int h) {
     int i=0;
  1: while (i < len && A[i] < h)
       i++;
  2: len = shift_array(A, i, 1);
     A[i] = h;
  3: while (i < len)
       i++;
  4: return i;
  }

predicates: $i_1 = i_2$, $i_1 < len_1$, $i_2 < len_2$,
     $A_1[i_1] < h_1$, $A_2[i_2] < h_2$, $len_1 = len_2$,
     $len_1 = len_2+1$, $len_2 = len_1+1$
\end{lstlisting}%
\end{minipage}
&
\begin{minipage}{0.4\textwidth}
  \lstset{basicstyle=\ttfamily\scriptsize}
  \lstset{language=C,morekeywords={assert,predicates,nd,composition}}
\begin{lstlisting}[mathescape=true]
composition:
if($pc_1 < 3$ && ($pc_2>0$ || !$cond_1$)
   && ($pc_2==3$||($pc_2==0$ && $cond_2$)))
       step(1);
else if ($pc_2<3$ && ($pc_1>0$ || !$cond_2$)
         && ($pc_1==3$ || ($pc_1==0$ && $cond_1$)))
            step(2);
else step(1,2);

       $cond_1$ := $i_1<len_1$ && $A_1[i_1]<h_1$
       $cond_2$ := $i_2<len_2$ && $A_2[i_2]<h_2$

\end{lstlisting}
\end{minipage}
\end{tabular}
\vspace{-0.1in}
  \caption{Constant-time insert to an array.}
  \label{fig:example}
\end{figure}

\vspace{-0.3cm}
\subsubsection{Related work.}
This paper addresses the problem of verifying k-safety properties (also called hyperproperties~\cite{DBLP:conf/csfw/ClarksonS08}) by means of self composition.
Other approaches tackle the problem without self-composition, and often focus on more specific properties, most noticeably the $2$-safety noninterference property (e.g.~\cite{DBLP:conf/pldi/AntonopoulosGHK17,DBLP:conf/cav/YangVSGM18}).
Below we focus on works that use self-composition. %

Previous work such as~\cite{DBLP:conf/csfw/BartheDR04,DBLP:conf/fm/BartheCK11,DBLP:conf/lfcs/BartheCK13,DBLP:conf/kbse/FelsingGKRU14,DBLP:conf/sas/TerauchiA05,DBLP:conf/esop/EilersMH18} considered self composition (also called product programs) where the composition function is constant and set a-priori, using syntax-based hints. While useful in general,
such self compositions may sometimes result in programs that are too complex to verify.
This is in contrast to our approach, where the composition function is evolving during verification,
and is adapted to the capabilities of the model checker.

The work most closely related to ours is~\cite{DBLP:conf/pldi/SousaD16} which introduces Cartesian Hoare Logic (CHL) for verification
of $k$-safety properties, and designs a verification framework for this logic.
This work is further improved in~\cite{DBLP:conf/cav/PickFG18}.
These works %
search for a proof in CHL, and in doing so, implicitly modify the composition.
Our work infers the composition explicitly %
and can use
off-the-shelf model checking tools. More importantly, when loops
are involved both~\cite{DBLP:conf/pldi/SousaD16} and~\cite{DBLP:conf/cav/PickFG18} use lock-step
composition and align loops syntactically. Our algorithm, in contrast, does
not rely on syntactic similarities, and can handle loops that cannot be aligned trivially.

There have been several results in the context of harnessing Constraint Horn Clauses (CHC) solvers for verification of relational properties~\cite{DBLP:conf/sas/AngelisFPP16,DBLP:conf/lpar/MordvinovF17}. 
Given several copies of a CHC system, %
a product CHC system that synchronizes the different copies is created by a syntactical analysis of the rules in the CHC system. 
These works restrict the synchronization points to CHC predicates (i.e., program locations), and consider only one synchronization (obtained via transformations of the system of CHCs). 
On the other hand, our algorithm iteratively searches for a good synchronization (composition), and considers synchronizations that depend on program state.

\paragraph{Equivalence checking and regression verification.}
Equivalence checking is another closely related research field, where a composition of several
programs is considered. As an example, equivalence checking is applied to verify the correctness
of compiler optimizations~\cite{DBLP:conf/fm/ZaksP08,DBLP:conf/oopsla/0001SCA13,DBLP:conf/aplas/DahiyaB17,DBLP:conf/sat/GuptaSMB18}. In~\cite{DBLP:conf/oopsla/0001SCA13}
the composition is determined by a brute-force search for possible synchronization points.
While this brute-force search resembles our approach for finding the correct composition,
it is not guided by the verification process.
The works in~\cite{DBLP:conf/aplas/DahiyaB17,DBLP:conf/sat/GuptaSMB18} identify
possible synchronization points syntactically, and try to match them
during the construction of a %
simulation relation between programs.

Regression verification also requires the ability to show equivalence between
different versions of a program~\cite{DBLP:conf/kbse/FelsingGKRU14,DBLP:conf/dac/GodlinS09,DBLP:conf/fm/StrichmanV16}.
The problem of synchronizing unbalanced loops appears
in~\cite{DBLP:conf/fm/StrichmanV16} in the form of unbalanced recursive function calls.
To allow synchronization in such cases, the user can specify different unrolling
parameters for the different copies. In contrast, our approach relies only on
user supplied predicates that are needed to establish correctness, while synchronization
is handled automatically.

\section{Preliminaries}

In this paper we reason about programs by means of the transition systems defining their semantics.
A transition system is a tuple $\TS = (\States, \Tr,\Terminal)$, where $\States$ is a set of states, $\Tr \subseteq \States \times \States$ is a transition relation that specifies the steps in an execution of the program, and $\Terminal \subseteq \States$ is a set  of \emph{terminal states} $\Terminal \subseteq \States$ such that every terminal state $\state \in \Terminal$ has an outgoing transition to itself and no additional transitions (terminal states allow us to reason about pre/post specifications of programs).
An \emph{execution} or \emph{trace} $\pi = s_0,s_1,\ldots$ is a (finite or infinite) sequence of states such that for every $i \geq 0$, $(s_i,s_{i+1}) \in \Tr$.
The execution is \emph{terminating} if there exists $0 \leq i \leq |\pi|$ such that $s_i \in \Terminal$. In this case, the suffix of the execution is of the form $s_i, s_i,\ldots$ and we say that $\pi$ ends at $s_i$.

As usual, we represent transition systems using logical formulas over a set of variables, corresponding to the program variables.
We denote the set of variables by $\Vars$. The set of terminal states is represented by a formula over $\Vars$ and the transition relation  is represented by a formula over $\Vars \uplus \Vars'$, where $\Vars$ represents the pre-state of a transition and $\Vars' = \{v' \mid v \in \Vars\}$ represents its post-state.
In the sequel, we use sets of states and their symbolic representation via formulas interchangeably.

\paragraph{Safety and inductive invariants.}
We consider safety properties defined via pre/post conditions.\footnote{Our results can be extended to arbitrary safety (and $k$-safety) properties by introducing ``observable'' states to which the property may refer.}
A \emph{safety property} is a pair $\prop{\preS}{\postS}$ where $\preS, \postS$ are formulas over $\Vars$, representing subsets of $\States$, denoting the pre- and post-condition, respectively.
$\TS$ \emph{satisfies} $\prop{\preS}{\postS}$, denoted $\TS \models \prop{\preS}{\postS}$, if every terminating execution $\pi$ of $\TS$ that starts in a state $s_0$ such that $s_0 \models \preS$ ends in a state $s$ such that $s \models \postS$.
In other words, for every state $s$ that is reachable in $\TS$ from a state in $\preS$
we have that %
$s \models \Terminal \rightarrow \postS$.

A prominent way to verify safety properties is by finding an inductive invariant. An \emph{inductive invariant} for a transition system $\TS$ and a safety property $\prop{\preS}{\postS}$ is a formula $\Inv$ such that
\begin{inparaenum}
\item[(1)] $\preS \Rightarrow \Inv$ (initiation),
\item[(2)] $\Inv \wedge \Tr \Rightarrow \Inv'$ (consecution), and
\item[(3)] $\Inv \Rightarrow (\Terminal \rightarrow \postS)$ (safety),
\end{inparaenum}
where $\varphi \Rightarrow \psi$ denotes the validity of $\varphi \to \psi$, and $\varphi'$ denotes $\varphi(\Vars')$, i.e., the formula obtained after substituting every $v \in \Vars$ by the corresponding $v' \in \Vars$.
If there exists such an inductive invariant, then $\TS \models \prop{\preS}{\postS}$.

\paragraph{$k$-safety.}
A \emph{$k$-safety property} refers to $k$ interacting executions of $\TS$. Similarly to an ordinary property, it is defined by $\prop{\preS}{\postS}$, except that $\preS$ and $\postS$ are defined over $\Vars^1 \uplus \ldots \uplus \Vars^k$ where $\Vars^i = \{v^i \mid v \in \Vars\}$ denotes the $i$th copy of the program variables. As such, $\preS$ and $\postS$ represent sets of $k$-tuples of program states (\emph{$k$-states} for short): for a $k$-tuple $(s_1,\ldots,s_k)$ of states and a formula $\varphi$ over $\Vars^1 \uplus \ldots \uplus \Vars^k$, we say that $(s_1,\ldots,s_k) \models \varphi$ if $\varphi$ is satisfied when for each $i$, the assignment of $\Vars^i$ is determined by $s_i$.
We say that $\TS$ \emph{satisfies} $\prop{\preS}{\postS}$, denoted $\TS \kmodels \prop{\preS}{\postS}$, if for every $k$ terminating executions $\pi^1,\ldots,\pi^k$ of $\TS$ that start in states $s_1,\ldots,s_k$, respectively, such that $(s_1,\ldots,s_k) \models \preS$, it holds that they end in states $t_1,\ldots,t_k$, respectively, such that $(t_1,\ldots,t_k) \models \postS$.

For example, the \emph{non interference} property may be specified by the
following $2$-safety property:
\ifextended
\[
\preS \;=  \bigwedge_{\var \in \mathrm{LowIn}} \var^1 = \var^2   \qquad  \qquad  \postS \;=\; \bigwedge_{\var \in \mathrm{LowOut}} \var^1 = \var^2  %
\]
\fi
where $\mathrm{LowIn}$ and $\mathrm{LowOut}$ denote subsets of the program inputs, resp. outputs, that are considered
``low security'' and the rest are classified as ``high security''.
This property asserts that every $2$ terminating executions that start in states
that agree on the ``low security'' inputs end in states that
agree on the low security outputs, i.e., the outcome does not
depend on any ``high security'' input and, hence, does not leak
secure information.

Checking $k$-safety properties reduces to checking ordinary safety properties by
creating a \emph{self composed program} that consists of $k$ copies of the
transition system, each with its own copy of the variables, that run in parallel
in some way. Thus, the self composed program is defined over variables
$\Varscomp = \Vars^1 \uplus \ldots \uplus \Vars^k$, where $\Vars^i = \{v^i \mid v \in \Vars\}$ denotes the variables associated with the $i$th copy.
For example, a common composition is a \emph{lock-step} composition
in which the copies execute simultaneously. The resulting composed transition system
$\TScompnof = (\Scompnof, \Tcompnof, \Fcompnof)$ is defined such that $\Scompnof
= \States \times \ldots \times \States$, $\Fcompnof = \bigwedge_{i=1}^k
\Terminal(\Vars^i)$ and $\Tcompnof = \bigwedge_{i=1}^k \Tr(\Vars^j,
{\Vars^j}')$. Note that $\Tcompnof$ is defined over $\Varscomp \uplus
{\Varscomp}'$ (as usual). Then, the $k$-safety property $\prop{\preS}{\postS}$
is satisfied by $\TS$ if and only if an ordinary safety property
$\prop{\preS}{\postS}$ is satisfied by $\TScompnof$. More general notions of
\emph{self composition} are investigated in \Cref{sec:problem}.

\section{Inferring Self Compositions for Restricted Languages of Inductive  Invariants} \label{sec:problem}

Any self-composition is sufficient for reducing $k$-safety to safety, e.g.,
lock-step, sequential, synchronous, asynchronous, etc. However, the choice of
the self-composition used determines the difficulty of the resulting safety
problem. Different self composed programs would require different inductive
invariants, some of which cannot be expressed in a given logical language.

In this section, we formulate the problem of inferring a self composition function such that the obtained self composed program may be verified with a given language of inductive invariants. We are, therefore, interested in inferring both the self composition function and the inductive invariant for verifying the resulting self composed program.
We start by formulating the kind of self compositions that we consider.

In the sequel, we fix a transition system $\TS = (\States, \Trans, \Terminal)$ with a set of variables $\Vars$.

\subsection{Semantic Self Composition}
\label{sec:SC-func}
Roughly speaking, a $k$ self composition of $\TS$ consists of $k$ copies of $\TS$ that execute together in some order, where steps may interleave or be performed simultaneously.
The order is determined by a self composition function, which may also be viewed as a scheduler that is responsible for scheduling a subset of the copies in each step. We consider \emph{semantic} compositions in which the order may depend on the \emph{states} of the different copies, as well as the correlations between them (as opposed to \emph{syntactic} compositions that only depend on the control locations of the copies, but may not depend on the values of other variables):

\begin{definition}[Semantic Self Composition Function]
A \emph{semantic $k$ self composition function} ($k$-composition function for short) is a function  $\SCfunc: \States^k \to \powerset{\{1..k\}}$, mapping each $k$-state to %
a \emph{nonempty} set of copies that are to participate in the next step of the
self composed program\footnote{We consider \emph{memoryless} composition
  functions. Compositions that depend on the history of the (joint) execution
  are supported via  ghost state added to the program to track the history.}.
\end{definition}

We represent a $k$-composition function
$\SCfunc$ by a set of logical conditions, with a condition $\Cond_M$ for every nonempty subset $M  \subseteq \{1..k\}$  of the copies. %
For each such $M \subseteq \{1..k\}$, the condition $\Cond_M$ is defined over $\Varscomp = \Vars^1 \uplus \ldots \uplus \Vars^k$, and hence it represents a set of $k$-states, with the meaning that all the $k$-states that satisfy $\Cond_M$ are mapped to $M$ by $\SCfunc$:
\vspace{-0.3cm}
\[\vspace{-0.2cm}
\SCfunc(\state_1,\ldots,\state_k) = M  \ \mbox{ if and only if } \ (\state_1,\ldots,\state_k) \models \Cond_{M}.
\]
To ensure that the function is well defined, we require that $(\bigvee_{M} \Cond_M) \equiv \true$, which ensures that every $k$-state satisfies at least one of the conditions. We also require that for every $M_1 \neq M_2$, $\Cond_{M_1} \wedge \Cond_{M_2} \equiv \false$, hence every $k$-state satisfies at most one condition. Together these requirements ensure that the conditions induce a partition of the set of all $k$-states.
In the sequel, we identify a $k$-composition function $\SCfunc$ with its
symbolic representation via  conditions $\{\Cond_M\}_M$ and use them
interchangeably.

\begin{definition}[Composed Program]
\label{def:comp-prog-semantic} \label{def:comp-program-logic}
Given a $k$-composition function $\SCfunc$, represented via conditions $\Cond_M$ for every nonempty set $M \subseteq \{1..k\}$, we define the \emph{$k$ self composition} of $\TS$ to be the transition system
$\TScomp = (\Scomp, \Tcomp,\Fcomp)$ over variables $\Varscomp= \Vars^1 \uplus \ldots \uplus \Vars^k$ defined as follows:  $\Fcomp = \bigwedge_{i=1}^k \Terminal^i$, where $\Terminal^i = \Terminal(\Vars^i)$, and
\vspace{-0.2cm}
\[
\vspace{-0.2cm}
\Tcomp = \bigvee_{\emptyset \neq M \subseteq \{1..k\}} \left(\Cond_M \wedge \varphi_M \right) \quad \mbox{  where } \quad \varphi_M = \bigwedge_{j \in M} \Tr(\Vars^j, {\Vars^j}') \wedge \bigwedge_{j \not\in M} \Vars^j = {\Vars^j}'
\]
\end{definition}
Thus, in $\TScomp$, the set of states consists of $k$-states ($\Scomp = \States \times\ldots\times \States$), the terminal states are $k$-states in which all the individual states are terminal, and the transition relation includes a transition from $(\state_1,\ldots, \state_k)$ to $(\state_1',\ldots, \state_k')$ if and only if $\SCfunc(\state_1,\ldots, \state_k) = M$ and
\ifextended
\begin{align*}
(\forall i\in M. \ (\state_i, \state_i') \in \Trans) \wedge (\forall i\not\in M.\ \state_i = \state_i')
\end{align*}
\fi
That is, every transition of $\TScomp$ corresponds to a simultaneous transition of a subset $M$ of the $k$ copies of $\TS$, where the subset is determined by the self composition function $\SCfunc$. If $\SCfunc(\state_1,\ldots,\state_k) = M$, then for every $i \in M$ we say that $i$ is \emph{scheduled} in $(\state_1,\ldots,\state_k)$.

\begin{example}
A $k$ self composition that runs the $k$ copies of $\TS$ sequentially, one after the other, corresponds to a $k$-composition function $\SCfunc$ defined by $\SCfunc(\state_1,\ldots,\state_k) = \{i\}$ where $i \in \{1..k\}$ is the minimal index of a non-terminal state in $\{\state_1,\ldots,\state_k\}$. If all states in $\{\state_1,\ldots,\state_k\}$ are terminal then $i =k$ (or any other index).
This is encoded as follows: for every $1 \leq i <k$, $\Cond_{\{i\}} = \neg \Terminal^i \wedge \bigwedge_{j<i} \Terminal^j$,
$\Cond_{\{k\}} = \bigwedge_{j<k} \Terminal^j$ and $\Cond_M = \false$ for every other $M \subseteq \{1..k\}$.
\end{example}

\begin{example}
The lock-step composition that runs the $k$ copies of $\TS$ %
synchronously
corresponds to a $k$-self composition function $\SCfunc$ defined by $\SCfunc(\state_1,\ldots,\state_k) = \{1,\ldots,k\}$, and encoded by $\Cond_{\{1,\ldots,k\}} = \true$ and $\Cond_M = \false$ for every other $M \subseteq \{1..k\}$.
\end{example}

In order to ensure soundness of a reduction of $k$-safety to safety via self composition, one has to require that the self composition function does not ``starve'' any copy of the transition system that is about to terminate if it continues to execute.
We refer to this requirement as \emph{fairness}.

\begin{definition}[Fairness]
A $k$-self composition function $\SCfunc$ is \emph{fair} if for every $k$ terminating executions $\pi^1,\ldots,\pi^k$ of $\TS$ there exists an execution $\picomp$ of $\TScomp$ such that for every copy $i \in \{1..k\}$, the projection of $\picomp$ to $i$ is $\pi^i$.
\end{definition}
Note that by the definition of the terminal states of $\TScomp$, $\picomp$ as above is guaranteed to be terminating.
We say that the $i$th copy \emph{terminates} in  $\picomp$ %
if $\picomp$ contains a $k$-state $(\state_1,\ldots,\state_k)$ such that
$\state_i \in \Terminal$.
Fairness may be enforced in a straightforward way by requiring that %
whenever $\SCfunc(\state_1,\ldots,\state_k)=M$, the set $M$ includes no index $i$ for which $\state_i \in \Terminal$, unless all have terminated.
Since we assume that terminal states may only transition to themselves, a weaker requirement that suffices to ensure fairness is that
$M$ includes at least one index $i$ for which $\state_i \not \in \Terminal$, unless there is no such index.

The following claim is now straightforward:

\begin{lemma} \label{thm:soundness}
Let $\TS$ be a transition system, $\prop{\preS}{\postS}$ a $k$-safety property, and $\SCfunc$ a fair $k$-composition function for $\TS$ and $\prop{\preS}{\postS}$. Then
\[
\TS \kmodels \prop{\preS}{\postS} \mbox{\ iff \ \ } \TScomp \models \prop{\preS}{\postS}.
\]
\end{lemma}
\begin{proof}[sketch]
Every terminating execution of $\TScomp$ corresponds to $k$ terminating executions of $\TS$. Fairness of $\SCfunc$ ensures that the converse also holds.
\end{proof}

To demonstrate the necessity of the fairness requirement, consider a (non-fair) self composition function $\SCfunc$ that maps every state to $\{1\}$. Then, regardless of what the actual transition system $\TS$ does, the resulting self composition $\TScomp$ satisfies every pre-post specification vacuously, as it never reaches a terminal state.

\begin{remark}
\label{rem:SC-rel}
While we require the conditions $\{\Cond_M\}_M$ defining a self composition function $\SCfunc$ to induce a partition of $\Scomp$ in order to ensure that $\SCfunc$ is well defined as a (total) function, the requirement may be relaxed in two ways. First, we may allow $\Cond_{M_1}$ and $\Cond_{M_2}$ to overlap. This will add more transitions and may make the task of verifying the composed program more difficult, but it maintains the soundness of the reduction. Second, it suffices that the conditions cover the set of \emph{reachable states} of the composed program rather than the entire state space.
These relaxations do not damage soundness.
Technically, this means that $\SCfunc$ represented by the conditions is a relation rather than a function. We still refer to it as a function and write $\SCfunc(\state_1,\ldots,\state_k) = M$ to indicate that $(\state_1,\ldots,\state_k) \models \Cond_M$, not excluding the possibility that $(\state_1,\ldots,\state_k) \models M'$ for $M' \neq M$ as well.
We note that as long as the language used to describe compositions is closed under Boolean operations, we can always extract from the conditions $\{\Cond_M\}_M$ a function $\SCfunc'$. This is done as follows:
\ifextended
\begin{itemize}
\item To prevent the overlap between conditions, determine an arbitrary total order $<$ on the sets $M \subseteq \{1..k\}$ and set $\Cond_M' := \Cond_M \wedge \bigwedge_{N < M} \neg \Cond_N$.
\item To ensure that the conditions cover the entire state space, set $\Cond_{\{1..k\}}' := \Cond_{\{1..k\}}' \vee \neg(\bigvee_M \Cond_M)$.
\end{itemize}
\fi
It is easy to verify that $\SCfunc'$ defined by $\{\Cond'_M\}_M$ is a total self composition function and that if $\SCfunc$ is fair, then so is $\SCfunc'$.
\end{remark}

\subsection{The Problem of Inferring Self Composition with Inductive Invariant}

\Cref{thm:soundness} states the soundness of the reduction of $k$-safety to ordinary safety. Together with the ability to verify safety by means of an inductive invariant, this leads to a verification procedure. However, while soundness of the reduction holds for \emph{any} self composition, an inductive invariant in a given language may exist for the  composed program resulting from some compositions but not from others. We therefore consider the self composition function and the inductive invariant together, as a pair, leading to the following definition.

\begin{definition} \label{def:comp-inv-pair}
Let $\TS$ be a transition system and $\prop{\preS}{\postS}$ a $k$ safety property. For a formula $\Inv$ over $\Varscomp$ and a self composition function $\SCfunc$ represented by conditions $\{\Cond_M\}_M$, we say that $(\SCfunc,\Inv)$ is a \emph{composition-invariant} pair for $\TS$ and $\prop{\preS}{\postS}$  if the following conditions hold:
\begin{itemize}
\item $\preS \implies \Inv$  \ \ (initiation of $\Inv$),
\item for every $\emptyset\neq M \subseteq \{1..k\}$, $\Inv \wedge \Cond_M \wedge \varphi_M \implies \Inv'$ (consecution of $\Inv$ for $\Tcomp$),
\item $\Inv  \implies \big((\bigwedge_{j=1}^k \Terminal^j) \rightarrow \postS \big)$ \ \ (safety of $\Inv$),
\item $\Inv \implies \bigvee_M \Cond_M$ \ \ \ ($\SCfunc$ covers the reachable states),
\item for every $\emptyset \neq M \subseteq \{1..k\}$, $\Cond_M \wedge (\bigvee_{j=1}^k \neg \Terminal^j) \implies \bigvee_{j\in M} \neg \Terminal^j$  \ \ \ ($f$ is fair).
\end{itemize}
\end{definition}

As commented in \Cref{rem:SC-rel}, we relax the requirement that $(\bigvee_M \Cond_M) \equiv \true$ to $\Inv \implies \bigvee_M \Cond_M$, thus ensuring that the conditions cover all the reachable states.
Since the reachable states of $\TScomp$ are determined by $\{\Cond_M\}_M$ (which define $\SCfunc$), this reveals the interplay between the self composition function and the inductive invariant.
Furthermore, we do not require that  $\Cond_{M_1} \wedge \Cond_{M_2} \equiv \false$ for $M_1 \neq M_2$, hence a $k$-state may satisfy multiple conditions.
As explained earlier, these relaxations do not damage soundness.
Furthermore, if we construct from $\SCfunc$ a self composition function $\SCfunc'$ as described in \Cref{rem:SC-rel}, $\Inv$ would be an inductive invariant for $\TScompprime$ as well.

\begin{lemma}
If there exists a composition-invariant pair $(\SCfunc,\Inv)$ for $\TS$ and $\prop{\preS}{\postS}$, then $\TS \kmodels \prop{\preS}{\postS}$.
\end{lemma}
\ifextended
\begin{proof}[sketch]
If $(\SCfunc,\Inv)$ is a composition-invariant pair, then $\Inv$ is an inductive invariant for $\TScompprime$, where $\SCfunc'$ is a fair composition function defined as in \Cref{rem:SC-rel}. From \Cref{thm:soundness} we conclude that $\TS \kmodels \prop{\preS}{\postS}$.
\end{proof}
\fi

If we do not restrict the language in which $\SCfunc$ and $\Inv$ are specified, then the converse also holds.
However, in the sequel we are interested in the ability to verify $k$-safety with a given language, e.g., one for which the conditions of \Cref{def:comp-inv-pair} belong to a decidable fragment of logic and hence can be discharged automatically.

\begin{definition}[Inference in $\lang$]
Let $\lang$ be a logical language. %
The problem of inferring a composition-invariant pair in $\lang$ is defined as follows. The input is a transition system $\TS$ and a $k$-safety property $(\preS, \postS)$.
The output is a composition-invariant pair $(\SCfunc, \Inv)$ for $\TS$ and $(\preS, \postS)$ (as defined in \Cref{def:comp-inv-pair}), where $\Inv \in \lang$ and $\SCfunc$ is represented by conditions $\{\Cond_M\}_M$ such that $\Cond_M \in \lang$ for every $\emptyset \neq M \subseteq \{1..k\}$. If no such pair exists, the output is ``no solution''.
\end{definition}
When no solution exists, it does not necessarily mean that $\TS \not \kmodels \prop{\preS}{\postS}$. Instead, it may be that the language $\lang$ is simply not expressive enough.
Unfortunately, for expressive languages (e.g., quantified formulas or even quantifier free linear integer arithmetic), the problem of inferring an inductive invariant alone is already undecidable, making the problem of inferring a composition-invariant pair undecidable as well:

\begin{lemma} \label{lem:undecidability}
Let $\lang$ be closed under Boolean operations and under substitution of a variable with a value, and include equalities of the form $v=a$, where $v$ is a variable and $a$ is a value (of the same sort).
If the problem of inferring an inductive invariant in $\lang$ is undecidable, then so is the problem of inferring a composition-invariant pair in $\lang$.
\end{lemma}

\iflong
\else
\ifextended
\begin{proof}\label{app:proof}
  We show a reduction from the ordinary invariant inference problem in $\lang$
  to the problem of inferring a composition-invariant pair in $\lang$. Given a
  transition system $\TS$ and an ordinary safety property $\prop{\preS}{\postS}$
  the reduction constructs a transition system $\TS^* = (\States^*, \Tr^*,
  \Terminal^*)$ over $\Vars^* = \Vars \uplus \{b\}$, where $b$ is a new Boolean
  variable such that when $b = \true$ the original transitions are taken and
  when $b = \false$ the systems remains in the same state, which is also added
  to the set of terminal states. Formally, for every $v \in \Vars$, let $a_v$
  be an arbitrary fixed value in the domain of $v$. For example, if $v$ is
  Boolean, $a_v = \false$.
    The reduction constructs
\begin{small}
\begin{align*}
\Tr^* &= (b \wedge \Tr \wedge b') \vee (\neg b \wedge (\bigwedge_{v\in \Vars} v' = a_v) \wedge \neg b')  &
\Terminal^*  &= \Terminal \vee (\neg b \wedge \bigwedge_{v\in \Vars} v' = a_v),
\end{align*}
\end{small}
and the following  $2$-safety property:
\begin{small}
\begin{align*} \preS^* &=  \left(  b^1 \wedge \preS(\Vars^1) \wedge \neg b^2 \wedge \bigwedge_{v\in \Vars} v^2 = a_v\right)  &
\postS^* &= \left(  b^1 \wedge \postS(\Vars^1) \wedge \neg b^2 \wedge \bigwedge_{v\in \Vars} v^2 = a_v\right) .
\end{align*}
\end{small}
That is, the first copy is ``initialized'' with $b = \true$ and with the original pre-condition and is required to terminate in a state that satisfies the original post-condition, while the second copy is initialized with $b =\false$, and with the value $a_v$ for each original variable, and is required to terminate in the same state. Clearly, if $\TS$ has an inductive invariant $\Inv$ for $\prop{\preS}{\postS}$, then $(\SCfunc, b^1 \wedge \Inv(\Vars^1)\wedge \neg b^2 \wedge \bigwedge_{v\in \Vars} v^2 = a_v)$ is a composition-invariant pair for $\TS^*$ and $\prop{\preS^*}{\postS^*}$, where $\SCfunc$ is defined by $\Cond_{\{1,2\}} = \true$ and $\Cond_M = \false$ for any other $M$, which is clearly in $\lang$. For the converse direction, if $\TS^*$ has a composition-invariant pair $(\SCfunc,\Inv^*)$ for $\prop{\preS^*}{\postS^*}$ then $\Inv$ obtained by substituting each positive occurrence of $b^2$ in $\Inv^*$ by $\false$, each negative occurrence of $b^2$ by $\true$ and each occurrence of $v^2$ by $a_v$ is an inductive invariant for $\TS$ and $\prop{\preS}{\postS}$.
\qed
\end{proof}
\fi
\fi
For example, linear integer arithmetic satisfies the conditions of the lemma.
This motivates us to restrict the languages of inductive invariants. Specifically, we consider languages defined by a finite set of predicates.
We consider \emph{relational} predicates, defined over $\Varscomp = \Vars^1 \uplus \ldots \uplus \Vars^k$.
For a finite set of predicates $\Preds$, we define $\langof{\Preds}$ to be the set of all formulas obtained by Boolean combinations of the predicates in $\Preds$.

\begin{definition}[Inference using predicate abstraction]
\label{def:inf-problem}
The problem of inferring a predicate-based composition-invariant pair is defined as follows. The input is a transition system $\TS$, a $k$-safety property $(\preS, \postS)$, and a finite set of predicates $\Preds$.
The output is the solution to the problem of inferring a composition-invariant pair for $\TS$ and $(\preS, \postS)$ in $\langof{\Preds}$.
\end{definition}

\begin{remark}
It is possible to decouple the language used for expressing the self composition function from the language used to express the inductive invariant.
Clearly, different sets of predicates (and hence languages) can be assigned to the self composition function and to the inductive invariant. However,
since inductiveness is defined with respect to the transitions of the composed system, which are in turn defined by the self composition function,
if the language defining $\SCfunc$ is not included in the language defining $\Inv$, the conditions $\Cond_M$ themselves would be over-approximated when checking the requirements of \Cref{def:comp-inv-pair} and therefore would incur a precision loss.
For this reason, we use the same language for both.
\end{remark}

Since the problem of invariant inference in $\langof{\Preds}$ is PSPACE-hard~\cite{DBLP:conf/cade/LahiriQ09},
a  reduction from the problem of inferring inductive invariants to the problem of inferring composition-invariant pairs (similar to the one used in the proof of \Cref{lem:undecidability}) shows
that  composition-invariant inference in $\langof{\Preds}$ is also PSPACE-hard:
\begin{theorem}
Inferring a predicate-based composition-invariant pair is PSPACE-hard.
\end{theorem}

\section{Algorithm for Inferring Composition-Invariant Pairs}

\SetKwFunction{fixFunc}{Modify\_SC}
\SetKwFunction{BMC}{BMC}
\SetKwFunction{safe}{Abs\_Reach}
\SetKwFunction{lastStep}{Last\_Step}
\SetKwFunction{Generalize}{Generalize}
\SetKwFunction{noValue}{All\_Excluded\_Or\_Starving}
\SetKwFunction{remLast}{Remove\_Last\_Step}
\begin{figure}[t]
\begin{algorithm}[H]
\footnotesize
\DontPrintSemicolon
\setcounter{AlgoLine}{0}
\LinesNumbered
\BlankLine
$\SCfunc \gets \text{lockstep}$ \label{line:lockstep}, $E \gets \emptyset$, $\Unreach \gets \false$ \;
\While { (true) } {
  $(res, \Inv, \cex) \gets \safe(\Preds,\TScomp,\preS,\postS,\Unreach)$ \label{line:checksafe} \; %
  \lIf {$res = \text{safe}$ }{
     \Return $(\SCfunc,\toPred{\Inv})$  \label{line:ret-true} %
   }
     $(\astate,M) \gets \lastStep(\cex)$ \;
     $E \gets E \cup \{(\astate,M)\}$ \label{line:add-e} \;
     \While {($\noValue(\astate, E)$) } {
        $\Unreach \gets \Unreach \vee \astate$ \label{line:add-unreach} \;
        \lIf{ $\Unreach \wedge \toBool{\varphi_{\preS}} \not\equiv \false$} {
           \Return \text{``no solution in $\langof{\Preds}$''} \label{line:ret-false} %
        }
           $\cex \gets \remLast(\cex)$ \;
           $(\astate,M) \gets \lastStep(\cex)$ \;
           $E \gets E \cup \{(\astate,M)\}$ \label{line:add-more-e} \;
     }
     $\SCfunc \gets \fixFunc(\SCfunc,\astate,E)$ \label{line:modify} \;
}
\caption{\Pdsc: Property-Directed Self-Composition. %
\label{alg:infer-2}}
\end{algorithm}
\end{figure}

In this section, we present %
Property Directed Self-Composition, \Pdsc for short --- our algorithm for tackling the composition-invariant inference problem for languages of predicates (\Cref{def:inf-problem}).
Namely, given a transition system $\TS$, a $k$-safety property $\prop{\preS}{\postS}$ and a finite set of predicates $\Preds$, we address the problem of finding a pair $(\SCfunc, \Inv$), where $\SCfunc$ is a self composition function and $\Inv$ is an inductive invariant for the composed transition system $\TScomp$ obtained from $\SCfunc$, and both of them are in $\langof{\Preds}$, i.e., defined by Boolean combinations of the predicates in $\Preds$.

We rely on the property that a transition system (in our case $\TScomp$) has an inductive invariant in $\langof{\Preds}$ if and only if its abstraction obtained using $\Preds$ is safe.
This is because, the set of reachable abstract states is the strongest set expressible in $\langof{\Preds}$ that satisfies initiation and consecution.
Given $\TScomp$, this allows us to use predicate abstraction to either obtain an inductive invariant in $\langof{\Preds}$ for $\TScomp$ (if the abstraction of $\TScomp$ is safe) or determine that no such inductive invariant exists (if an abstract counterexample trace is obtained). The latter indicates that a different self composition function needs to be considered.
A naive realization of this idea gives rise to an iterative algorithm that starts from an arbitrary initial composition function and in each iteration computes a new composition function. %
At the worst case such an algorithm enumerates all self composition functions defined in $\langof{\Preds}$, i.e., has time complexity $O(2^{2^{|\Preds|}})$. Importantly, we observe that, when no inductive invariant exists for some composition function, we can use the abstract counterexample trace returned in this case to (i)~generalize and eliminate multiple composition functions, and (ii)~identify that some abstract states must be unreachable if there is to be a composition-invariant pair, i.e., we ``block'' states in the spirit of \emph{property directed reachability}~\cite{ic3,pdr}.
This leads to the algorithm depicted in \Cref{alg:infer-2} whose worst case time complexity is
$2^{O(|\Preds|)}$.
Next, we explain the algorithm in detail.

\subsubsection*{Finding an inductive invariant for a given composition function using predicate abstraction.}
We use predicate abstraction~\cite{DBLP:conf/cav/GrafS97,DBLP:conf/cav/SaidiS99} to check if a given candidate composition function has a corresponding inductive invariant.
This is done as follows.
The abstraction of $\TScomp$ using $\Preds$, denoted $\A{\Preds}(\TScomp)$, is a transition system $(\AStates, \ATr)$ %
defined over variables $\mathcal{B}$, where $\mathcal{B} = \{b_\pred \mid \pred \in \Preds\}$ (we omit the terminal states).
$\AStates = \{0,1\}^{\mathcal{B}}$, i.e., each abstract state corresponds to a valuation of the Boolean variables representing $\Preds$.
An abstract state $\astate \in \AStates$ represents the following set of states of $\TScomp$:
\[
\gamma(\astate) = \{ \scomp \in \Scomp \mid \forall \pred \in \Preds. \ \scomp \models \pred \Leftrightarrow \astate(b_{\pred}) = 1\}
\]
We extend $\gamma$ to sets of states and to formulas representing sets of states in the usual way.
The abstract transition relation %
is defined as usual:
\[
\ATr = \{(\astate_1, \astate_2) \mid \exists \scomp_1 \in \gamma(\astate_1) \ \exists  \scomp_2 \in \gamma(\astate_2).\ (\scomp_1,\scomp_2) \in \Tcomp\}
\]
Note that the set of abstract states %
in $\A{\Preds}(\TScomp)$ does \emph{not} depend on $\SCfunc$.

\begin{notation}
We sometimes refer to an abstract state $\astate  \in \AStates$ as the formula $\bigwedge_{\astate(b_{\pred}) = 1}  b_{\pred} \wedge
\bigwedge_{\astate(b_{\pred}) = 0} \neg b_{\pred}$.
For a formula $\psi \in \langof{\Preds}$, we denote by $\toBool{\psi}$ the result of substituting each $\pred \in \Preds$ in $\psi$ by the corresponding Boolean variable $b_{\pred}$.
For the opposite direction, given a formula $\psi$ over $\mathcal{B}$, we denote by $\toPred{\psi}$ the formula in $\langof{\Preds}$ resulting from substituting each $b_{\pred} \in \mathcal{B}$ in $\psi$ by $\pred$. Therefore, $\toPred{\psi}$ is a symbolic representation of $\gamma(\psi)$.
\end{notation}

Every set defined by a formula $\psi \in \langof{\Preds}$ is precisely represented by $\toBool{\psi}$ in the sense that
$\gamma(\toBool{\psi})$ is equal to the set of states defined by $\psi$, i.e., $\toBool{\psi}$ is a precise abstraction of $\psi$.
For simplicity, we assume that the termination conditions as well as the pre/post specification can be expressed precisely using the abstraction,
in the following sense:

\begin{definition}
$\Preds$ is \emph{adequate} for $\TS$ and $\prop{\preS}{\postS}$ if there exist $\varphi_{\preS}, \varphi_{\postS}, \varphi_{\Terminal^i} \in \langof{\Preds}$ such that $\varphi_{\preS} \equiv \preS$, $\varphi_{\postS} \equiv \postS$ and $\varphi_{\Terminal^i} \equiv \Terminal^i$ (for every copy $i \in \{1..k\}$).
\end{definition}

\ignore{
Adequacy ensures that $\Fcomp = \bigwedge_{i=1}^k \Terminal^i$
is precisely represented by $\bigwedge_{i=1}^k \toBool{\varphi_{\Terminal^i}}$. %
This implies that (the symbolic representation of) $\AF$ can be defined by $\bigwedge_{i=1}^k \toBool{\varphi_{\Terminal^i}}$ and $\gamma(\AF) = \Fcomp$.
Together with the precise representation of $\preS$ and $\postS$, this ensures that no precision loss is incurred by the initiation and safety requirements of \Cref{def:comp-inv-pair}, and they can be established using the abstraction.
Recall that, in addition, the predicates used to define $\SCfunc$ are the same as those used to construct the abstraction, i.e., each condition $\Cond_M$ in the definition of $\SCfunc$ is in $\langof{\Preds}$.
Thus, the requirements on $\SCfunc$ can also be established using the abstraction.
We therefore get the following:
}

The following lemma provides the foundation for our algorithm:

\begin{lemma}
\label{lem:abs-check}
Let $\TS$ be a transition system, $\prop{\preS}{\postS}$ a $k$ safety property, and $\Preds$ a finite set of predicates adequate for $\TS$ and $\prop{\preS}{\postS}$.
For a self composition function $\SCfunc$ defined via conditions $\{\Cond_M\}_M$ in $\langof{\Preds}$, there exists an inductive invariant $\Inv$ in $\langof{\Preds}$ such that $(\SCfunc,\Inv)$ is a composition-invariant pair for $\TS$ and $\prop{\preS}{\postS}$ if and only if the following three conditions hold:
\begin{compactitem} %
\item[{\bf S1}] All reachable states of $\A{\Preds}(\TScomp)$ from $\toBool{\varphi_{\preS}}$ satisfy $(\bigwedge_{i=1}^k \toBool{\varphi_{\Terminal^i}})\rightarrow \toBool{\varphi_{\postS}}$,
\item[{\bf S2}] All reachable states of $\A{\Preds}(\TScomp)$  from $\toBool{\varphi_{\preS}}$ satisfy $\bigvee_M \toBool{\Cond_M}$, and
\item[{\bf S3}] For every $\emptyset\neq M \subseteq \{1..k\}$, $\toBool{\Cond_M} \wedge (\bigvee_{j=1}^k \neg \toBool{\varphi_{\Terminal^j}})\implies \bigvee_{j\in M} \neg \toBool{\varphi_{\Terminal^j}}$.
\end{compactitem}
Furthermore, if the conditions hold, then the symbolic representation of the set of abstract states of $\A{\Preds}(\TScomp)$ reachable from $\toBool{\varphi_{\preS}}$ is a formula $\Inv$ over $\mathcal{B}$ such that $(\SCfunc, \toPred{\Inv})$ is a composition-invariant pair for $\TS$ and $\prop{\preS}{\postS}$.
\end{lemma}
\ifextended
\begin{proof}
The proof relies on the following statement, denoted by $(*)$: for a formula $\varphi$ in $\langof{\Preds}$ and an abstract state $\astate$, for every $\scomp\in\gamma(\astate)$ it holds that $\scomp\models\varphi \Leftrightarrow \astate\models \toBool{\varphi}$
(which follows by induction on the structure of a formula in $\langof{\Preds}$, relying on the definition of $\gamma(\astate)$). In particular, this implies that for a formula $\psi$ over $\mathcal{B}$, it holds that $\scomp\models\toPred{\psi} \Leftrightarrow \astate\models {\psi}$ whenever $\scomp\in\gamma(\astate)$.

($\Rightarrow$) Let $\TS$, $\prop{\preS}{\postS}$ and $\Preds$ be as described, and let ($\SCfunc,\Inv$) be a composition-invariant pair for $\TS$ and $\prop{\preS}{\postS}$ in $\langof{\Preds}$. We first show that every (abstract) state that is reachable from $\toBool{\varphi_{\preS}}$ in $\A{\Preds}(\TScomp)$ satisfies $\toBool{\Inv}$. Let $\astate$ be such a reachable state. Then there exists an abstract trace $\astate_1,\ldots,\astate_m$ such that $\astate_1\models\toBool{\varphi_{\preS}}$, $\astate_m = \astate$ and $(\astate_i,\astate_{i+1})\in \ATr$ for every $1 \leq i < m$. Consider a concrete state $\scomp_1$ of $\TScomp$ such that $\scomp_1 \in \gamma(\astate_1)$, then $\astate_1\models\toBool{\varphi_{\preS}}$ and from $(*)$ we get $\scomp_1\models \varphi_{\preS}$. From the definition of a composition-invariant pair (\Cref{def:comp-inv-pair}) we get that  $\scomp_1\models\Inv$ (initiation). Since $\Inv$ is in $\langof{\Preds}$ we get from $(*)$ that also $\astate_1\models \toBool{\Inv}$. For $\astate_2$, the next state in the abstract trace, it also holds that $\astate_2\models \toBool{\Inv}$: since $(\astate_1,\astate_2)\in\ATr$, we know that there exist some $\scomp_a\in \gamma(\astate_1)$ and  $\scomp_b\in \gamma(\astate_2)$ such that $(\scomp_a,\scomp_b) \in \Tcomp$, using $(*)$ we get that $\scomp_a\models Inv$, the consecution of $\Inv$ implies $\scomp_b\models\Inv$ and from $(*)$ we get $\astate_2\models \toBool{\Inv}$. By induction over the length of the abstract trace we get that $\astate\models \toBool{\Inv}$. We now turn to show that conditions {\bf S1--S3} hold.
First, the safety of $\Inv$ for $\TScomp$ together with adequacy of $\Preds$ and $(*)$ imply that $\toBool{\Inv} \implies \big((\bigwedge_{j=1}^k \toBool{\Terminal^j}) \rightarrow \toBool{\postS} \big)$, and since all the reachable states of $\A{\Preds}(\TScomp)$ satisfy $\toBool{\Inv}$,  {\bf S1} follows.
Similarly, the covering requirement of $\SCfunc$ together with the property that $C_M$ %
is in $\langof{\Preds}$ for every $M$ and together with $(*)$ imply {\bf S2}.
Finally, {\bf S3} is implied directly from the fairness of $\SCfunc$ (\Cref{def:comp-inv-pair}).

($\Leftarrow$) Assume that for $\TS$, $\prop{\preS}{\postS}$, $\Preds$ and some composition function $\SCfunc$ as described, conditions {\bf S1--S3} hold.
Condition {\bf S1} ensures that $\A{\Preds}(\TScomp)$ satisfies the safety property $(\toBool{\varphi_{\preS}},\toBool{\varphi_{\postS}} )$, when we augment $\A{\Preds}(\TScomp)$ with a set of terminal states given by the formula $\bigwedge_{i=1}^k \toBool{\varphi_{\Terminal^i}}$.
Hence, there exists an inductive invariant $\Inv$ over $\mathcal{B}$ for $\A{\Preds}(\TScomp)$ and $(\toBool{\varphi_{\preS}},\toBool{\varphi_{\postS}} )$.
Furthermore, condition {\bf S2} ensures that there exists such $\Inv$ for which $\Inv \implies \bigvee_M \toBool{\Cond_M}$ (for example, such $\Inv$ may be obtained by conjoining the inductive invariant ensured by {\bf S1} with another inductive invariant that establishes {\bf S2}). %
To conclude the proof we show that ($\SCfunc,\toPred{\Inv}$) is a composition-invariant pair for $\TS$ and $\prop{\preS}{\postS}$, as defined in \Cref{def:comp-inv-pair}. %
First, initiation and safety of $\Inv$ with respect to $\A{\Preds}(\TScomp)$ and $\prop{\toBool{\varphi_{\preS}}}{\toBool{\varphi_{\postS}}}$, imply initiation and safety (respectively) of $\toPred{\Inv}$ with respect to $\TS$ and $\prop{\varphi_{\preS}}{\varphi_{\postS}}$ due to $(*)$ and adequacy of $\Preds$.
As for consecution of $\toPred{\Inv}$: for a pair of states $\scomp_1,\scomp_2$ in $\TScomp$ such that $(\scomp_1,\scomp_2)\in \Tcomp$, if $\scomp_1\in\gamma(\astate_1)$ and $\scomp_2\in\gamma(\astate_2)$, then $(\astate_1,\astate_2)\in \ATr$. Therefore, if $\scomp_1\models\toPred{\Inv}$ then $\astate_1\models {\Inv}$ (according to $(*)$), and from consecution of ${\Inv}$ in $\A{\Preds}(\TScomp)$ also $\astate_2\models {\Inv}$, and from $(*)$ we get $\scomp_2\models \toPred{\Inv}$ and conclude the consecution of $\toPred{\Inv}$ in $\TScomp$.
Similarly, for covering of $\SCfunc$: recall that $\Inv \implies \bigvee_M \toBool{\Cond_M}$, hence by $(*)$,
$\toPred{\Inv} \implies \bigvee_M {\Cond_M}$, i.e., $\SCfunc$ covers the states satisfying $\toPred{\Inv}$. Finally, the fairness of $\SCfunc$ follows from {\bf S3}.
\qed

\end{proof}
\fi

\Cref{alg:infer-2} starts from the lock-step self composition function
(\Cref{line:lockstep}), which is fair\footnote{Any fair self composition can be
  chosen as the initial one; we chose lock-step since it is a good starting point in many applications.},
  and constructs the next candidate
$\SCfunc$ such that condition {\bf S3} in \Cref{lem:abs-check} always holds (see
discussion of $\fixFunc$).
Thus, %
condition {\bf S3} need not be checked explicitly.

\Cref{alg:infer-2} checks whether conditions {\bf S1} and {\bf S2} hold for a given candidate composition function $\SCfunc$ by calling $\safe$ (\Cref{line:checksafe}) -- both checks are performed via a \mbox{(non-)reachability} check in $\A{\Preds}(\TScomp)$, checking whether a state violating $(\bigwedge_{i=1}^k \toBool{\varphi_{\Terminal^i}})\rightarrow \toBool{\varphi_{\postS}}$ or $\bigvee_M \toBool{\Cond_M}$ is reachable from $\toBool{\varphi_{\preS}}$.
\Cref{alg:infer-2} maintains the abstract states that are not in $\bigvee_M \toBool{\Cond_M}$ by the formula $\Unreach$ defined over $\mathcal{B}$, which is initialized to $\false$ (as  the lock-step composition function is defined for every state) and is updated in each iteration of \Cref{alg:infer-2} to include the abstract states violating $\bigvee_M \toBool{\Cond_M}$.
If no abstract state violating {\bf S1} or {\bf S2} is reachable, i.e., the conditions hold, then $\safe$ returns the (potentially overapproximated) set of reachable abstract states, represented by a formula $\Inv$ over $\mathcal{B}$.  In this case, by \Cref{lem:abs-check}, $(\SCfunc, \toPred{\Inv})$ is a composition-invariant pair (\Cref{line:ret-true}). Otherwise, an abstract counterexample trace is obtained.
(We can of course apply bounded model checking to check if the counterexample is real; we omit this check as our focus is on the case where the system is safe.)

\ignore{
If the trace exhibits violation of {\bf S1}, we use $\BMC$ on $\TScomp$ to determine whether the counterexample is real (\Cref{line:bmc}). If it is, the system violates its $k$ safety property, thus \Cref{alg:infer-2} terminates.
Otherwise, %
the search for a composition-invariant pair continues.
}

\begin{remark}
In practice, we do not construct $\A{\Preds}(\TScomp)$ explicitly.  Instead, we use the \emph{implicit predicate abstraction} approach~\cite{DBLP:conf/tacas/CimattiGMT14}.
\end{remark}

\subsubsection*{Eliminating self composition candidates based on abstract counterexamples.}

An abstract counterexample to conditions {\bf S1} or {\bf S2} indicates that the candidate composition function $\SCfunc$ has no corresponding $\Inv$.
Violation of {\bf S1} can only be resolved by changing $\SCfunc$ such that the abstract trace is no longer feasible.
Violation of {\bf S2} may, in principle, also be resolved by extending the definition of $\SCfunc$ such that it is defined for all the abstract states in the counterexample trace.

However, to prevent the need to explore both options,
our algorithm maintains the following invariant for every candidate self composition function $\SCfunc$ that it constructs:
\begin{claim}
Every abstract state that is \emph{not} in $\bigvee_M \toBool{\Cond_M}$ is not reachable w.r.t. the abstract composed program of \emph{any} composition function that is part of a composition-invariant pair for $\TS$ and $\prop{\preS}{\postS}$.
\end{claim}
This property clearly holds for the lock-step composition function, which the algorithm starts with, since for this composition, $\bigvee_M \toBool{\Cond_M} \equiv \true$.
As we explain in \Cref{cor:unreach}, it continues to hold throughout the algorithm.

As a result of this property, whenever a candidate composition function $\SCfunc$ does not satisfy condition {\bf S1} or {\bf S2},
it is never the case that
$\bigvee_M \toBool{\Cond_M}$ needs to be extended to allow the abstract states in $\cex$ to be reachable.
Instead, the abstract counterexample obtained in violation of the conditions needs to be
eliminated %
by modifying $\SCfunc$. %

Let $\cex = \astate_1,\ldots,\astate_{m+1}$ be an abstract counterexample of $\A{\Preds}(\TScomp)$  such that $\astate_1 \models \toBool{\varphi_{\preS}}$ and $\astate_{m+1} \models (\bigwedge_{i=1}^k \toBool{\varphi_{\Terminal^i}})\wedge \neg \toBool{\varphi_{\postS}}$ (violating {\bf S1}) or $\astate_{m+1} \models  \Unreach$ (violating {\bf S2}).
Any self composition $\SCfunc'$ that agrees with $\SCfunc$ on the states in
$\gamma(\astate_i)$ for every $\astate_i$ that appears in $\cex$ has the
same transitions in $\Tcomp$ and, hence, the same transitions in $\ATr$. It,
therefore, exhibits the same abstract counterexample in
$\A{\Preds}(\TScompprime)$. Hence, it violates {\bf S1} or {\bf S2} and is not part of any composition-invariant pair.

\begin{notation}
Recall that $\SCfunc$ is defined via conditions $\Cond_M \in \langof{\Preds}$.
This ensures that for every abstract state $\astate$, $\SCfunc$ is defined in the same way for all the states in $\gamma(\astate)$.
We %
denote %
the value of $\SCfunc$ on the states in $\gamma(\astate)$ by $\SCfunc(\astate)$ (in particular, $\SCfunc(\astate)$ may be undefined).
We get that $\SCfunc(\astate) = M$ if and only if $\astate \models \toBool{\Cond_M}$.
\end{notation}
Using this notation, to eliminate the abstract counterexample $\cex$, one needs to eliminate at least one of the transitions in $\cex$ by changing the definition of $\SCfunc(\astate_i)$ for \emph{some} $1 \leq i \leq m$.
For a new candidate function $\SCfunc'$ this may be encoded by the disjunctive constraint $\bigvee_{i=1}^m \SCfunc'(\astate_i) \neq \SCfunc(\astate_i)$.
However, we observe that a stronger requirement may be derived from $\cex$ based on the following lemma:

\begin{lemma} \label{lem:diff-value}
Let $\SCfunc$ be a self composition function and $\cex = \astate_1,\ldots,\astate_{m+1}$ a counterexample trace in $\A{\Preds}(\TScomp)$ %
such that $\astate_1 \models \toBool{\varphi_{\preS}}$ but $\astate_{m+1} \models (\bigwedge_{i=1}^k \toBool{\varphi_{\Terminal^i}})\wedge \neg \toBool{\varphi_{\postS}}$ or $\astate_{m+1} \models \Unreach$.
Then for any self composition function $\SCfunc'$ such that $\SCfunc'(\astate_{m}) = \SCfunc(\astate_{m})$, if $\astate_{m}$ is reachable in $\A{\Preds}(\TScompprime)$ from $\toBool{\varphi_{\preS}}$, then a counterexample trace to {\rm {\bf S1}} or {\rm {\bf S2}} exists. %
\end{lemma}
\ifextended
\begin{proof}
Suppose that $\astate_{m}$ is reachable in $\A{\Preds}(\TScompprime)$ from $\toBool{\varphi_{\preS}}$. Then there exists a trace $\astate'_1,\ldots,\astate'_{m}$ in $\A{\Preds}(\TScompprime)$ such that $\astate'_1\models\toBool{\varphi_{\preS}}$ and $\astate'_m = \astate_m$.
Since $\SCfunc'(\astate_{m}) = \SCfunc(\astate_{m})$, the outgoing transitions of $\astate_{m}$ are the same in both $\A{\Preds}(\TScomp)$ and $\A{\Preds}(\TScompprime)$.
In particular, the transition $(\astate_m, \astate_{m+1})$ from $\A{\Preds}(\TScomp)$ also exists in $\A{\Preds}(\TScompprime)$.
Therefore, $\cex'=\astate'_1,\ldots,\astate'_{m},\astate_{m+1}$ is a trace to $\astate_{m+1}$ in $\A{\Preds}(\TScompprime)$. %
If $\astate_{m+1} \models (\bigwedge_{i=1}^k \toBool{\varphi_{\Terminal^i}})\wedge \neg \toBool{\varphi_{\postS}}$, then $\cex'$ is a counterexample to {\rm {\bf S1}} in $\A{\Preds}(\TScompprime)$ as well.
Consider the case where $\astate_{m+1} \models \Unreach$. %
By the construction of $\Unreach$, this %
indicates that $\astate_{m+1}$ has an outgoing abstract trace that leads to violation of {\rm {\bf S1}} or {\rm {\bf S2}} with every non-starving self composition function, and in particular in $\A{\Preds}(\TScompprime)$.
\qed
\end{proof}
\fi
\begin{corollary} \label{cor:diff-value}
If there exists a composition-invariant pair $(\SCfunc', \Inv')$, then there is also one where $\SCfunc'(\astate_{m}) \neq \SCfunc(\astate_{m})$.
\end{corollary}
\ifextended
\begin{proof}
If $\SCfunc'(\astate_{m}) = \SCfunc(\astate_{m})$, then by \Cref{lem:diff-value}, $\astate_{m}$ is necessarily unreachable in $\A{\Preds}(\TScompprime)$ from $\toBool{\varphi_{\preS}}$. Therefore, if we change $\SCfunc'(\astate_{m})$, all the requirements of \Cref{lem:abs-check} will still hold.
If no alternative value that admits the fairness requirement exists, then $\SCfunc'(\astate_m)$ can remain undefined.
\qed
\end{proof}
\fi

Therefore, we require that in the next self composition candidates the abstract state $\astate_m$ must not be mapped to its current value in $\SCfunc$, i.e.,  %
$\SCfunc'(\astate_{m}) \neq M$, where $\SCfunc(\astate_{m}) = M$\footnote{If the conditions $\{\Cond_M\}_M$ defining $\SCfunc$ may overlap, we consider the condition $\Cond_M$ %
by which the transition from $\astate_{m}$ to $\astate_{m+1}$ was defined.}.

\Cref{alg:infer-2} accumulates these constraints in the set $E$ (\Cref{line:add-e}). %
Formally, the constraint $(\astate, M) \in E$ asserts that $\Cond_M'$ must imply $\neg (\bigwedge_{\astate(b_{\pred}) = 1}  {\pred} \wedge
\bigwedge_{\astate(b_{\pred}) = 0} \neg {\pred})$, and hence $\SCfunc'(\astate) \neq M$.

\subsubsection*{Identifying abstract states that must be unreachable.}
A new candidate self composition is constructed such that it satisfies all the constraints in $E$ (thus ensuring that no abstract counterexample will re-appear).
In the construction, we make sure to satisfy {\bf S3} %
(fairness). Therefore, for every abstract state $\astate$, %
we choose a value $\SCfunc'(\astate)$ that satisfies the constraints in $E$ and is \emph{non-starving}: a value $M$ is starving for $\astate$ if $\astate \models \bigvee_{j=1}^k \neg \toBool{\varphi_{\Terminal^j}}$ but $\astate \not \models \bigvee_{j\in M} \neg \toBool{\varphi_{\Terminal^j}}$, i.e., some of the copies have not terminated in $\astate$ but none of the non-terminating copies is scheduled. (Due to adequacy, a value $M$ is starving for $\astate$ if and only if it is starving for every $\scomp \in \gamma(\astate)$.)

If for some abstract state $\astate$, all the non-starving values have already been excluded (i.e., $(\astate,M) \in E$ for every non-starving $M$), we %
conclude that there is \emph{no} $\SCfunc'$ such that $\astate$ is reachable in $\A{\Preds}(\TScompprime)$ and $\SCfunc'$ is part of a composition-invariant pair:
\begin{lemma}
\label{lem:unreach}
Let $\astate \in \AStates$ be an abstract state such that for every $\emptyset \neq M \subseteq \{1..k\}$ either $M$ is starving for $\astate$ or $(\astate,M) \in E$.
Then, for every $\SCfunc'$ that satisfies {\rm {\bf S3}}, if $\A{\Preds}(\TScompprime)$ satisfies {\rm {\bf S1}} and {\rm {\bf S2}}, then $\astate$ is unreachable in $\A{\Preds}(\TScompprime)$.
\end{lemma}
\ifextended
\begin{proof}
If $\SCfunc'$ satisfies {\bf S3} and $\A{\Preds}(\TScompprime)$ satisfies {\bf S1} and {\bf S2}, then according to \Cref{lem:abs-check} $\SCfunc'$ is a part of some composition-invariant pair $(\SCfunc',\Inv)$ for $\TS$.
Furthermore, as shown in the proof of \Cref{lem:abs-check}, every (abstract) state that is reachable from $\toBool{\varphi_{\preS}}$ in $\A{\Preds}(\TScompprime)$ satisfies $\toBool{\Inv}$.
Assume to the contrary that $\astate$ is reachable in $\A{\Preds}(\TScompprime)$.
Then $\astate \models \toBool{\Inv}$.
According to \Cref{def:comp-inv-pair}, $\SCfunc'$ must be defined for $\astate$, thus $\SCfunc'(\astate)=M'$ for some $\emptyset\neq M'\subseteq\{1\ldots k\}$. Since $\SCfunc'$ is fair (satisfies {\bf S3}) it must be the case that $(\astate,M')\in E$. According to the algorithm, at some iteration there was a composition $\SCfunc''$ with $\SCfunc''(\astate)=M'$ that caused adding $(\astate,M')$ to $E$, i.e., there was a counterexample to {\bf S1} or {\bf S2} in $\A{\Preds}(\TScompdoubleprime)$ in the form of a trace to $\astate$.
Then \Cref{lem:diff-value} implies that there is also a counterexample to {\bf S1} or {\bf S2} in $\A{\Preds}(\TScompprime)$ because $\SCfunc'(\astate)=\SCfunc''(\astate)=M'$. This contradicts the assumption that $\A{\Preds}(\TScompprime)$ satisfies {\bf S1} and {\bf S2}.
\qed
\end{proof}
\fi
\begin{corollary} \label{cor:unreach}
If there exists a composition-invariant pair $(\SCfunc', \Inv')$, then $\astate$ is unreachable in $\A{\Preds}(\TScompprime)$.
\end{corollary}
This is because no matter how the self composition function $\SCfunc'$ would be defined, $\astate$ is guaranteed to have an outgoing abstract counterexample trace in $\A{\Preds}(\TScompprime)$.

We, therefore, turn $\SCfunc'(\astate)$ to be undefined.
As a result, condition {\bf S2} of \Cref{lem:abs-check} requires that $\astate$ will be unreachable in $\A{\Preds}(\TScompprime)$.
In \Cref{alg:infer-2}, this is enforced by adding $\astate$ %
to $\Unreach$ (\Cref{line:add-unreach}).

Every abstract state $\astate$ that is added to $\Unreach$ is a strengthening of
the safety property by an additional constraint that needs to be obeyed in any composition-invariant pair, where obtaining a composition-invariant  pair is the target of the algorithm.
This makes our algorithm \emph{property directed}.

If an abstract state that satisfies $\toBool{\varphi_{\preS}}$ is added to $\Unreach$, then \Cref{alg:infer-2} determines that no solution exists (\Cref{line:ret-false}).
Otherwise, it generates a new constraint for $E$ based on the abstract state preceding $\astate$ in the abstract counterexample (\Cref{line:add-more-e}).

\subsubsection*{Constructing the next candidate self composition function.}
Given the set of constraints in $E$ and the formula $\Unreach$, $\fixFunc$ (\Cref{line:modify}) generates the next candidate composition function by (i)~taking a constraint $(\astate, M)$ such that $\astate \not \models \Unreach$ (typically the one that was added last), (ii)~selecting a non-starving value $M_{\text{new}}$ for $\astate$ (such a value must exist, otherwise $\astate$ would have been added to $\Unreach$), and (iii)~updating the conditions defining $\SCfunc'$ as follows:
\vspace{-0.3cm}
\begin{align*}
\Cond_M'  &= \Cond_M \wedge \neg \toPred{\astate}&
\Cond_{M_{\text{new}}}' &= \left(  \Cond_{M_{\text{new}}} \vee \toPred{\astate} \right)
\vspace{-0.3cm}
\end{align*}
The conditions of other values remain as before.
This definition is facilitated by the fact that the same set of predicates is used both for defining $\SCfunc'$ and for defining the abstract states $\astate \in \AStates$ (by which $\Inv$ is obtained).
Note that in practice we do not explicitly turn $\SCfunc'$ to be undefined for $\gamma(\Unreach)$. However, these definitions are ignored.
The definition ensures that $\SCfunc'$ is non-starving (satisfying condition {\bf S3}) and that no two conditions $\Cond'_{M_1}\neq \Cond'_{M_2}$ overlap.
While the latter is not required, it also does not restrict the generality of the approach (since the language we consider is closed under Boolean operations).

\begin{theorem}
Let $\TS$ be a transition system, $\prop{\preS}{\postS}$ a $k$-safety property and $\Preds$ a set of predicates over $\Varscomp$.
If \Cref{alg:infer-2} returns `` no solution'' then there is no composition-invariant pair for $\TS$ and $\prop{\preS}{\postS}$ in $\langof{\Preds}$.
Otherwise, $(\SCfunc, \toPred{\Inv})$ returned by \Cref{alg:infer-2} is a composition-invariant pair in $\langof{\Preds}$, and thus $\TS \kmodels \prop{\preS}{\postS}$.
\end{theorem}
\ifextended
\begin{proof}
\Cref{alg:infer-2} returns `` no solution'' when $\Unreach \wedge \toBool{\varphi_{\preS}}$ is satisfiable. This means that there is an abstract state $\astate$ that satisfies $\toBool{\varphi_{\preS}}$ but also satisfies $\Unreach$. By the construction of $\Unreach$, this means that $\astate$ must be unreachable from $\toBool{\varphi_{\preS}}$ in any $\A{\Preds}(\TScompprime)$ such that $(\SCfunc', \Inv')$ a composition-invariant pair in $\langof{\Preds}$ (see \Cref{cor:unreach}). Hence, no such $(\SCfunc', \Inv')$ exists.
Conversely, \Cref{alg:infer-2} returns $(\SCfunc, \toPred{\Inv})$ when all the conditions listed in \Cref{lem:abs-check} are met, thus $(\SCfunc, \toPred{\Inv})$ is a composition-invariant pair.
\qed
\end{proof}
\fi

\paragraph{Complexity.}
Each iteration of \Cref{alg:infer-2} adds at least one constraint to $E$, excluding a potential value for $\SCfunc$ over some abstract state $\astate$.
An excluded values is never re-used.
Hence, the number of iterations is at most the number of abstract states, $2^{|\Preds|}$, multiplied by the number of potential values for each abstract state, $n=2^k$.
Altogether, the number of iterations is at most $O(2^{|\Preds|} \cdot 2^k)$.
Each iteration makes one call to $\safe$ which checks reachability via predicate abstraction, hence, assuming that satisfiability checks in the original logic are at most exponential, its complexity is  $2^{O(|\Preds|)}$.
Therefore, the overall complexity of the algorithm is
$2^{O(|\Preds|)+k}$.
Typically, $k$ is a small constant, hence the complexity is dominated by %
$2^{O(|\Preds|)}$.

\section{Evaluation and Conclusion}

\paragraph{Implementation. }
We implemented \Pdsc (\Cref{alg:infer-2}) 
in Python on top of Z3~\cite{DBLP:conf/tacas/MouraB08}.
Its input is %
a transition system encoded by %
Constrained Horn Clauses (CHC) in SMT2 format, a $k$-safety property and
a set of predicates. The abstraction is implicitly encoded using the
approach of~\cite{DBLP:conf/tacas/CimattiGMT14}, and is parameterized
by a composition function that is modified in each iteration.
For reachability checks ($\safe$) we use \Spacer~\cite{DBLP:conf/cav/KomuravelliGC14}, which supports LRA and arrays.
For the set of predicates used by \Pdsc,
we implemented an automatic procedure that mines these
predicates from the CHC. Additional predicates may be added manually.

\paragraph{Experiments. }
To evaluate \Pdsc, we compare it to \Syn~\cite{DBLP:conf/cav/PickFG18}, the current state of the art in $k$-safety verification.

To show the effectiveness of \Pdsc, we consider
examples that require a \emph{nontrivial} composition
\ifextended
(these examples are detailed in \Cref{app:examples}).
\fi
We emphasize that the motivation for these
example is originated in real-life scenarios. For example, \Cref{fig:example} follows
a pattern of constant-time execution. %
The results of these experiments are summarized in \Cref{tbl:results}.
\Pdsc is able to find the right composition function and prove all
of the examples, while \Syn cannot verify any of them.
We emphasize that for these examples, lock-step composition is not sufficient.
However, \Pdsc infers %
a composition that depends on
the programs' state (variable values), rather than just program locations.

\begin{figure}[t]
\begin{floatrow}

\capbtabbox{%
\scalebox{0.9}{
  \begin{tabular} {|| c ||c  | c | c || c ||}
    \hline
   \multirow{2}{*}{\textbf{Program}} & \multicolumn{2}{c|}{PDSC} & \multirow{2}{*}{SYNONYM} \\
    \cline{2-3}
     	     & Time(s)			& Iteations	& \\
    \hline
    DoubleSquareNI 	     & 7			& 33	& fail	\\
    \hline
    HalfSquareNI          & 3.4			& 28	& fail	\\
	\hline
	ArrayIntMod  & 58.2			& 168	& fail \\
	\hline
	SquaresSum   & 2.8			& 4		& fail \\
	\hline
	ArrayInsert  & 19.5			& 102	& fail \\
	\hline	
\end{tabular}
}
}{%
  \caption{Examples that require semantic compositions}\label{tbl:results}
}
\ffigbox{%
  \includegraphics[width=0.3\textwidth]{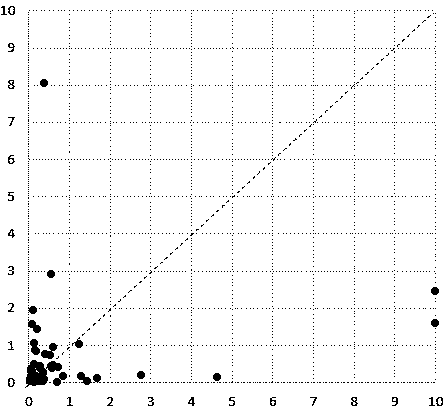}
}{%
  \caption{Runtime comparison (in sec.): \Pdsc (x-axis) and \Syn (y-axis).}\label{fig:syn}%
}
\end{floatrow}
\vspace{-0.4cm}
\end{figure}

Next we consider Java programs
from~\cite{DBLP:conf/cav/PickFG18,DBLP:conf/pldi/SousaD16},
which we manually converted to C, and then converted to CHC using \SeaHorn~\cite{DBLP:conf/cav/GurfinkelKKN15}. 
For all but 3 examples, only 2 types of predicates, which we mined automatically, were sufficient
for verification:
(i)~relational predicates derived from the pre- and post-conditions, %
and (ii)~for simple loops that have an index variable (e.g., for iterating over an array), an equality predicate between the copies of the indices.
These predicates were sufficient %
since we used a large-step encoding
of the transition relation, hence the abstraction via predicates takes effect only at cut-points.
For the remaining 3 examples, we manually added 2--4 predicates.
With the exception of 1 example where a timeout of 10 seconds was reached, all examples 
were solved with a lock-step composition function.
Yet, we include them
to show that on  examples with simple compositions  \Pdsc performs similarly to \Syn.
This can be seen in \Cref{fig:syn}.

%
\vspace{-0.3cm}
\subsubsection{Conclusion and Future Work.}
This work formulates the problem of inferring a self composition function together with an inductive invariant for the composed program,
thus capturing the interplay between the self composition and the difficulty of verifying the resulting composed program.
To address this problem we present \Pdsc -- an algorithm for inferring a semantic self composition, directed at verifying the composed program with a given language of predicates.
We show that \Pdsc manages to find nontrivial self compositions that are beyond reach of existing tools.
In future work, we are interested in further improving \Pdsc by extending it with additional (possibly lazy) predicate discovery abilities. This has the potential to both improve performance and verify properties over wider range of programs.
Additionally, we consider exploring further generalization techniques during the inference procedure.

\vspace{-0.3cm}
\subsubsection*{Acknowledgements}
This publication is part of a project that has received funding from the European Research Council (ERC) under the European Union's Horizon 2020 research and innovation programme (grant agreement No [759102-SVIS]).
The research was partially supported by
Len Blavatnik and the Blavatnik Family foundation,
the Blavatnik Interdisciplinary Cyber Research Center, Tel Aviv University,
the Israel Science Foundation (ISF) under grant No. 1810/18
and the United States-Israel Binational Science Foundation (BSF) grant No. 2016260.
 
\bibliographystyle{splncs04}
\bibliography{refs}
\ifextended
\newpage
\appendix
\section{Benchmarks Used in the Evaluation}
\label{app:examples}

In this section, we elaborate on the examples from \Cref{tbl:results}.

\subsection{DoubleSquareNI}
\begin{figure}[H]
\footnotesize
\raggedright
\begin{tabular}{cc}
\begin{minipage}{0.3\textwidth}
\lstset{basicstyle=\ttfamily\scriptsize}
\lstset{language=C,morekeywords={post,pre,bool,assert,predicates,nd,composition}}
\begin{lstlisting}[mathescape=true]
pre($x_1 == x_2$)

doubleSquare(bool h, int x){
  int z, y=0;
  if(h) { z = 2*x; }
  else { z = x; }
  while (z>0) {
  	z--;
  	y = y+x;
  }
  if(!h) { y = 2*y; }
  return y;
}

post($y_1 == y_2$)
\end{lstlisting}
\end{minipage}
& \hspace{0.8cm}
\begin{minipage}{0.4\textwidth}
\lstset{basicstyle=\ttfamily\scriptsize}
\lstset{language=C,morekeywords={predicates,composition}}
\begin{lstlisting}[mathescape=true]
predicates:
  $h_1$, $h_2$, $x_1>0$, $y_1\geq 0$, $y_2\geq 0$, $z_2\geq 0$,
  $z_2\geq 0$, $x_1=x_2$, $y_1=y_2$, $y_1=2y_2$, $y_2=2y_1$,
  $z_1=z_2$, $z_1=2z_2$, $z_2=2z_1$, $z_1=2z_2-1$,
  $z_2=2z_1-1$, $y_1=2y_2+x_2$, $y_2=2y_1+x_1$
\end{lstlisting}
\vspace{0.5cm}
\begin{small}
\lstset{basicstyle=\ttfamily\scriptsize}
\lstset{language=C,morekeywords={predicates,composition}}
\begin{lstlisting}[mathescape=true]
composition =
if((($z0>0$ & $z_2>0$ | ($z_1\leq 0$ & $z_2\leq 0$))
  && ($h_1$ & $z_1==2*z_2$)
  && !($h_1==h_2$ || ($z_1==0$ & $z_2==0$)))
    || (!($z_1>0$ & $z_2>0$ | ($z_1\leq 0$ & $z_2\leq 0$))
      & $z_2\leq 0$ & $z_1>0$))
    step (1);
else if ((($z_1>0$ & $z_2>0$ | ($z_1\leq 0$ & $z_2\leq 0$))
        && !($h_1==h_2$ | ($z_1==0$ & $z_2==0$))
        && !($h_1$ & $z_1==2*z_2$) & ($z_2==2*z_1$))
          || !($z_1>0$ & $z_2>0$ | ($z_1\leq 0$ & $z_2\leq 0$)))
	step (2);
else
    step(1,2);

\end{lstlisting}
\end{small}
\end{minipage}
\end{tabular}
\caption{A program that computes $2x^2$; the computation depends on a secret bit $h$ while $x$ is the low input. \label{fig:example_power}}
\end{figure}

\Cref{fig:example_power} depicts a non-interference problem (a 2-safety problem) where $x$ is the low input and $h$ is the high input. Taint analysis methods cannot prove non-interference for this program, and no proof exists when the product program presented in~\cite{DBLP:conf/esop/EilersMH18} is applied (see \Cref{sec:insufficieny}). However, using the language of predicates presented (also in \Cref{fig:example_power}), our algorithm infers a composition-invariant pair that proves non-interference for the program.

\subsection{HalfSquareNI}
\begin{figure}[H]
\footnotesize
\centering
\begin{tabular}{cc}
\begin{minipage}{0.4\textwidth}
\lstset{basicstyle=\ttfamily\scriptsize}
\lstset{language=C,morekeywords={pre,post,predicates,composition}}
\begin{lstlisting}[mathescape=true]
pre($low_1 == low_2$)

halfSquare(int h, int low){
  assume(low > h > 0);
  int i = 0, y = 0, v = 0
  while (h > i) {
  	i++; y += y;
  }
  v = 1;
  while (low > i) {
  	i++; y += y;
  }
  return y;
}

post($y_1 == y_2$)
\end{lstlisting}
\end{minipage}
& \hspace{1cm}
\begin{minipage}{0.4\textwidth}
\lstset{basicstyle=\ttfamily\scriptsize}
\lstset{language=C,morekeywords={predicates,composition}}
\begin{lstlisting}[mathescape=true]
predicates:
  $h_1>0$, $h_2>0$ , $low_1>h_1$,
  $low_2>h_2$, $i_1<h_1$ ,$i_2<h_2$,
  $i_1<low_1$, $i_2<low_2$, $v_1=1$,
  $v_2=1$, $y_1=y_2$, $i_1=i_2$,
  $low_1=low_2$
\end{lstlisting}
\vspace{0.5cm}
\lstset{basicstyle=\ttfamily\scriptsize}
\lstset{language=C,morekeywords={predicates,composition}}
\begin{lstlisting}[mathescape=true]
composition =
  if (($v_1 == 0$ && $i_1 \geq h_1$)
  	& ($i_2 < h_2$ || $v_2 == 1$))
      step(1);
  else if (($v_2 \neq 1$ && $i_2 \geq h_2$)
  	  & ($i_1 < h_1$ || $v_1 == 1$))
      step(2);
  else
      step(1,2);
\end{lstlisting}
\end{minipage}
\end{tabular}
\caption{A program that computes $\frac{low^2}{2}$; the computation is not continuous and depends on a secret variable $h$. \label{fig:example_sum}}
\end{figure}

In the program presented at \Cref{fig:example_sum} we consider the non-interference property, with pre-condition $low_1=low_2$ (low input) and post-condition $y_1=y_2$ (non-interference). The high input $h$ has no constraints as implied from the pre-condition. Intuitively, the difficulty of proving non-interference for this program is the need to "skip" the statement between the two loops in order to keep the outputs of the copies equal in every composed state along the execution. The suggested composition aligns the computations such that they proceed simultaniously only when both are at either loops, which makes $i_1=i_2\wedge y_1=y_2$ true for every state of the self composed program.

\subsection{ArrayIntMod}
\begin{figure}[H]
\footnotesize
\centering
\begin{tabular}{cc}
\begin{minipage}{0.4\textwidth}
\lstset{basicstyle=\ttfamily\scriptsize}
\lstset{language=C,morekeywords={pre,post,boolean,predicates,composition}}
\begin{lstlisting}[mathescape=true]
pre($o1_1 == o2_2$ && $o2_1 == o1_2$)

int compare(AInt o1, AInt o2){
     if(o1.len != o2.len){
        return 0;
     }
     boolean flag = (o1.get(0)>0);
     int i, aentry, bentry,last1,last2;
     i = 0;

     while ((i < o1.len) && (i<o2.len)) {
       aentry = o1.get(i);
       bentry = o2.get(i);
       if (aentry < bentry) {
           return -1;
       }
       if (aentry > bentry) {
           return 1;
       }
       i++;

     if(flag && ((i < o1.len) && (i < o2.len))){
          aentry = o1.get(i);
          bentry = o2.get(i);
          if (aentry < bentry) {
           return -1;
            }
            if (aentry > bentry) {
           return 1;
            }
          i++;
       }

     }
     return 0;
   }
}

post(sign(compare($o1_1$,$o2_1$)) = -sign(compare($o1_2$,$o2_2$))

\end{lstlisting}
\end{minipage}
& \hspace{2cm}
\begin{minipage}[b]{0.4\textwidth}
\lstset{basicstyle=\ttfamily\scriptsize}
\lstset{language=C,morekeywords={predicates,composition}}
\begin{lstlisting}[mathescape=true]
predicates:
  $len1_1=len2_1$, $len1_1=len2_2$,
  $o1_1=o2_2$, $o2_1=o1_2$,
  $len1_1=len2_2$, $len2_1=len1_2$,
  $ret_1=-ret_2$, $i_1<len1_1$,
  $i_2<len1_2$, $i_1=i_2$,
  $i_1=i_2-1$, $i_2=i_1-1$,
  $o1_1[i_1]=o2_1[i_1]$,
  $o1_2[i_2]=o2_2[i_2]$,
  $o1_1[i_1+1]=o2_1[i_1+1]$,
  $o1_2[i_2+1]=o2_2[i_2+1]$
\end{lstlisting}
\vspace{0.5cm}
\end{minipage}
\end{tabular}

\begin{minipage}{\textwidth}
\footnotesize
\centering
\lstset{basicstyle=\ttfamily\scriptsize}
\lstset{language=C,morekeywords={predicates,composition}}
\begin{lstlisting}[mathescape=true]
composition =
  if (($i_2=i_1+1$ && $i_2<len_2$) || ($i_2=i_1$ && $i_2<len_2$ && $o1_2[i_2] = o2_2[i_2]$ &&
        $o1_2[i_2+1] \neq o2_2[i_2+1]$ && $flag_2$ && !$flag_1$)
      step(1);
  else if (($i_1=i_2+1$ && $i_1<len_1$) || ($i_1=i_2$ && $i_1<len_1$ && $o1_1[i_1] = o2_1[i_1]$ &&
            $o1_1[i_1+1] \neq o2_1[i_1+1]$ && $flag_1$ !$flag_2$))
      step(2);
  else
      step(1,2);
\end{lstlisting}
\end{minipage}
\caption{Comparator example with potentially unbalanced loops. \label{fig:example_arrayint}}
\end{figure}

The example in \Cref{fig:example_arrayint} is a comparator based on a Java comparator from the evaluation comparator programs. The comparator was modified to have loop that might perform two steps in a single iteration. The 2-safety property to prove for the comparator is anti-symmetry, i.e. the pre-condition is $o1_1 = o2_2 \wedge o1_2 = o2_1$ and the post-condition is $sign(compare(o1_1,o2_1)) = -sign(compare(o1_2,o2_2))$. %
The figure also describes a composition that aligns the loops according to the value of $flag$. This yields a composed program that has an invariant that proves the desired property in the predicates language from \Cref{fig:example_arrayint}.

\subsection{SquaresSum}
\begin{figure}[H]
\footnotesize
\centering
\begin{tabular}{cc}
\begin{minipage}{0.4\textwidth}
\lstset{basicstyle=\ttfamily\scriptsize}
\lstset{language=C,morekeywords={pre,post,predicates,composition}}
\begin{lstlisting}[mathescape=true]
pre($a_1<a_2$ && $b_2<b_1$)

squaresSum(int a, int b){
  assume(0 < a < b);
  int c=0;
  while (a<b) {c+=a*a; a++;}
  return c;
}

post($c_2<c_1$)
\end{lstlisting}
\end{minipage}
& \hspace{1cm}
\begin{minipage}{0.4\textwidth}
\lstset{basicstyle=\ttfamily\scriptsize}
\lstset{language=C,morekeywords={predicates,composition}}
\begin{lstlisting}[mathescape=true]
predicates:
  $c_1>c_2$, $c_1=c_2$, $a_1<a_2$
  $a_1=a_2$, $b_1>b_2$, $a_1<b_1$
  $a_2<b_2$, $b_1>1$, $b_2>1$
\end{lstlisting}
\vspace{0.5cm}
\lstset{basicstyle=\ttfamily\scriptsize}
\lstset{language=C,morekeywords={predicates,composition}}
\begin{lstlisting}[mathescape=true]
composition =
  if ($a_1 < a_2$)
  	step(1);
  else
  	step(1,2);
\end{lstlisting}
\end{minipage}
\end{tabular}
\caption{A program that computes $\mathlarger{\sum}_{a\leq n<b}{n^2}$. \label{fig:example_sqsum}} %
\end{figure}
For the program described in \Cref{fig:example_sqsum} we consider the monotonicity property - a 2-safety property with pre-condition $[a_2,b_2]\subset[a_1,b_1]$ and post-condition $c_2<c_1$. Considering a composition that aligns the computations to start together and run simultaniously, it is easy to see that $c_1<c_2$ for unbounded number iterations. However, in \Cref{fig:example_sqsum} we see a composition that eases the task of finding a proof by scheduling the copies such that $c_2<c_1$ holds from the first iteration of copy 2 and to the end of both computations.

\subsection{ArrayInsert}
The program with a detailed explanation of its proof using a composition-invariant pair are presented in \Cref{sec:intro}.
\section{Demonstrating the Interplay Between Self Composition and Inductive Invariants}
\label{sec:insufficieny}

We illustrate the effect of the self composition function on the difficulty of verifying the obtained composed program, as well as the need for a semantic self composition function %
on the simple example depicted in \Cref{fig:example_power}. The program receives as input an integer $x$ and a secret bit $h$, and outputs $y = 2x^2$.
The desired specification is that the output does not depend on $h$, which is indeed the case. Formally, this is a $2$-safety property with pre-condition $x_1=x_2$ and post-condition  $y_1 = y_2$, 
requiring that in any two terminating executions that start with the same values for $x$, the final value of $y$ is the same.

As explained earlier, any fair self composition function can be soundly used to reduce the $2$-safety problem to an ordinary safety problem.
This is because the variables of the two copies of the program are completely disjoint, making the states completely independent. Therefore, the output of each copy does not depend on the actual interleaving of the two copies. As a result, if some interleaving (a fair self composition function) violates the postcondition, all of them will. That is, the actual interleaving does not affect the soundness of the reduction to traditional safety.
However, when we turn to verifying the safety of the composed program by finding an inductive invariant in a given language, the specific self composition function used plays a significant role.
For example, consider a composition that ``synchronizes'' the two copies in each control structure (e.g.~\cite{DBLP:conf/esop/EilersMH18}). Such a composition runs the two copies of the loop in parallel until one copy exits the loop, and then continues to run the other copy. The self composed program obtained by this composition is displayed in \Cref{fig:example-sc-bad}. 

\begin{figure}[H]
\lstset{basicstyle=\ttfamily\scriptsize}
\lstset{language=C,morekeywords={predicates,composition}}
\begin{lstlisting}[mathescape=true]
doubleSquare_sc(bool h1, bool h2, int x){
  int z1, z2, y1 = 0; y2 = 0;
  if(h1) {z1 = 2*x;} else {z1 = x;}
  if(h2) {z2 = 2*x;} else {z2 = x;}
  while (z1>0 || z2>0) {
    if (z1>0) {z1--; y1 = y1+x;}
    if (z2>0) {z2--; y2 = y2+x;}
  }
  if(!h1) {y1 = 2*y1;}
  if(!h2) {y2 = 2*y2;}
  return y1, y2;
}
\end{lstlisting}
\caption{Self composition for the program depicted in \Cref{fig:example_power} based on~\cite{DBLP:conf/esop/EilersMH18}. \label{fig:example-sc-bad}}
\end{figure}

We show that for this composition, there exists \emph{no inductive invariant} in quantifier free linear integer arithmetic (QFLIA) that is sufficient for establishing safety of the composed program.

\begin{proof}
If we examine the set $R$ of the reachable states of the composed program at the exit point we see that it includes (for every natural $n$):
\begin{footnotesize}
\begin{align*}
(x, y_1, z_1, y_2, z_2) \mapsto \qquad & (n, 0, 2n, 0, n)\\
& (n, n, 2n-1, n, n-1)\\
& \ldots \\
& (n, kn, 2n-k, kn, n-k)\\
& \ldots \\
& (n, n^2, n, n^2, 0) \\
& (n, n^2+n, n-1, n^2, 0) \\
& \ldots \\
& (n, n^2+kn, n-k, n^2, 0)\\
& \ldots \\
& (n, 2n^2, 0, n^2, 0)\\
\end{align*}
\end{footnotesize}
(We omit the second copy of $x$ since both copies are equal in all the reachable states -- a fact that is also expressible in QFLIA -- and similarly, we omit $h_1$ and $h_2$.)

Clearly, an inductive invariant must be satisfied by all of these states, since all of them are reachable.
However, we show that any QFLIA formula that is satisfied by all of these states is also satisfied by a state that reaches a bad state (i.e., a state where $y_1 \neq y_2$), thus if it is safe, it necessarily violates the consecution requirement, which means it is not an inductive invariant.

Let $\varphi = \varphi_1 \lor \ldots \varphi_r$ be a QFLIA formula, written in DNF form, where each $\varphi_i$ is a cube (conjunction of literals). Define $R_1,\ldots,R_r \subseteq R$ such that $R_i = \{s \in R \mid s \models \varphi_i\}$ includes all states in $R$ that satisfy $\varphi_i$. We show that there exists $i$ such that $\varphi_i$ is also satisfied by a state that reaches a bad state.

$R$ includes infinitely many ``points'' of the form $n, n^2, n, n^2, 0$ where $n$ is an even number. Therefore, since there are finitely many $R_i$'s that together cover $R$, there exists $i$ such that $R_i$ also includes infinitely many such points. Take two such points $(n, n^2, n, n^2, 0)$ and $(m, m^2, m, m^2, 0)$  in $R_i$ where $n \neq m$. Then $(1/2(n+m), 1/2(n^2+m^2), 1/2(n+m), 1/2(n^2+m^2), 0)$ is a state (all values are integers) in the convex hull of $R_i$. In particular, it must satisfy $\varphi_i$ ($\varphi_i$ is a cube in LIA that is satisfied by all states in $R_i$, hence it is also satisfied by all states in its convex hull).

However, when executing the while loop starting from the state $x \mapsto 1/2(n+m), y_1 \mapsto 1/2(n^2+m^2), z_1 \mapsto 1/2(n+m), y_2 \mapsto 1/2(n^2+m^2), z_2 \mapsto 0$, the outcome is the state $x \mapsto 1/2(n+m), y_1 \mapsto 1/2(n^2+m^2)+1/4(n+m)^2, z_1 \mapsto 0, y_2 \mapsto 1/2(n^2+m^2), z_2 \mapsto 0$, where $y_1 \neq y_2$, hence safety is violated.

This means that $\varphi$ is not an inductive invariant strong enough to establish safety of the composed program, in contradiction.
\qed
\end{proof}

In contrast, with the composition function inferred by \Pdsc (see \Cref{fig:example_power}), the composed program has an inductive invariant in QFLIA.  \fi

\end{document}